\documentclass[11pt,letterpaper]{article}

\usepackage[utf8]{inputenc}
\usepackage[T1]{fontenc}

\usepackage{amsmath}
\usepackage{amssymb}
\usepackage{amsthm}
\usepackage{mathtools}
\usepackage{thmtools}
\usepackage{isomath}
\usepackage{dsfont}

\usepackage{enumerate}

\usepackage{libertine}
\usepackage[libertine]{newtxmath} 
\usepackage[scaled=0.96]{zi4} 

\usepackage[DIV=11]{typearea} 

\usepackage{microtype}

\usepackage[framemethod=TikZ]{mdframed}
\mdfsetup{%
    backgroundcolor=gray!10,
    leftmargin=10pt,
    rightmargin=10pt,
    middlelinewidth=0pt,
    roundcorner=5pt,
}

\usepackage[linesnumbered,vlined,ruled]{algorithm2e}

\usepackage{ifthen}
\usepackage{xcolor}
\usepackage{comment}

\usepackage{todonotes}
\usepackage{lineno}

\usepackage[
    maxalphanames=4,
    minalphanames=4,
    maxnames=20,
    minnames=20,
	backref=true,
	doi=true,
	url=true,
	backend=bibtex8,
]{biblatex}

\usepackage[ocgcolorlinks]{hyperref} 

\usepackage{tabularx}

\allowdisplaybreaks[1]

\makeatletter
\AtBeginDocument{%
    \newlength{\temp@x}%
    \newlength{\temp@y}%
    \newlength{\temp@w}%
    \newlength{\temp@h}%
    \def\my@coords#1#2#3#4{%
      \setlength{\temp@x}{#1}%
      \setlength{\temp@y}{#2}%
      \setlength{\temp@w}{#3}%
      \setlength{\temp@h}{#4}%
      \adjustlengths{}%
      \my@pdfliteral{\strip@pt\temp@x\space\strip@pt\temp@y\space\strip@pt\temp@w\space\strip@pt\temp@h\space re}}%
    \ifpdf
      \typeout{In PDF mode}%
      \def\my@pdfliteral#1{\pdfliteral page{#1}}
      \def\adjustlengths{}%
    \fi
    \ifxetex
      \def\my@pdfliteral #1{}
      \def\adjustlengths{\setlength{\temp@h}{-\temp@h}\addtolength{\temp@y}{1in}\addtolength{\temp@x}{-1in}}%
    \fi%
    \def\Hy@colorlink#1{%
      \begingroup
        \ifHy@ocgcolorlinks
          \def\Hy@ocgcolor{#1}%
          \my@pdfliteral{q}%
          \my@pdfliteral{7 Tr}
        \else
          \HyColor@UseColor#1%
        \fi
    }%
    \def\Hy@endcolorlink{%
      \ifHy@ocgcolorlinks%
        \my@pdfliteral{/OC/OCPrint BDC}%
        \my@coords{0pt}{0pt}{\pdfpagewidth}{\pdfpageheight}%
        \my@pdfliteral{F}
        \my@pdfliteral{EMC/OC/OCView BDC}%
        \begingroup%
          \expandafter\HyColor@UseColor\Hy@ocgcolor%
          \my@coords{0pt}{0pt}{\pdfpagewidth}{\pdfpageheight}%
          \my@pdfliteral{F}
        \endgroup%
        \my@pdfliteral{EMC}%
        \my@pdfliteral{0 Tr}
        \my@pdfliteral{Q}%
      \fi
      \endgroup
    }%
}
\makeatother

\colorlet{DarkRed}{red!50!black}
\colorlet{DarkGreen}{green!50!black}
\colorlet{DarkBlue}{blue!50!black}

\hypersetup{
	linkcolor = DarkRed,
	citecolor = DarkGreen,
	urlcolor = DarkBlue,
	bookmarksnumbered = true,
	linktocpage = true
}

\usepackage{lowco2}

\DontPrintSemicolon
\SetKwProg{Procedure}{Procedure}{}{}
\SetProcNameSty{textsc}
\SetFuncSty{textsc}
\SetKw{KwAnd}{and}
\SetKw{KwDo}{do}

\SetCommentSty{mycommfont}

\declaretheorem[numberwithin=section]{theorem}
\declaretheorem[numberlike=theorem]{lemma}

\declaretheorem[numberlike=theorem]{definition}
\declaretheorem[numberlike=theorem]{claim}

\allowdisplaybreaks

\DeclareMathOperator{\Vol}{Vol}
\DeclareMathOperator{\N}{\mathbb N}

\DeclareMathOperator{\E}{\mathbb E}
\renewcommand{\d}{\partial}
\renewcommand{\P}{\mathbb P}
\newcommand{\PPPP}{k}

\usepackage{xspace}

\newcommand{\pushpull}[1]{\texttt{#1-PUSH\&PULL}\xspace}

\usepackage[capitalize]{cleveref}
\crefname{claim}{claim}{claims}

\title{Towards Constant Time Multi-Call Rumor Spreading on Small-Set Expanders}

\author{
  Emilio Cruciani\thanks{European University of Rome, Italy}
  \and
  Sebastian Forster\thanks{Department of Computer Science, University of Salzburg, Austria}
  \and
  Tijn de Vos\thanks{TU Graz, Austria}
}

\date{}

\hypersetup{
	pdftitle = {Towards Constant Time Multi-Call Rumor Spreading on Small-Set Expanders},
	pdfauthor = {Emilio Cruciani, Sebastian Forster, Tijn de Vos}
}

\usepackage{environ}
\NewEnviron{hide}{}

\begin{document}

\begin{titlepage}
\maketitle

\begin{abstract}
We study a multi-call variant of the classic \texttt{PUSH\&PULL} rumor spreading process where nodes can contact $k$ of their neighbors instead of a single one during both \texttt{PUSH} and \texttt{PULL} operations.
We show that rumor spreading can be made faster at the cost of an increased amount of communication between the nodes. 
As a motivating example, consider the process on a complete graph of $n$ nodes: while the standard \texttt{PUSH\&PULL} protocol takes $\Theta(\log n)$ rounds, we prove that our $k$-\texttt{PUSH\&PULL} variant completes in $\Theta(\log_{k} n)$ rounds, with high probability.

We generalize this result in an expansion-sensitive way, as has been done for the classic \texttt{PUSH\&PULL} protocol for different notions of expansion, e.g., conductance and vertex expansion.
We consider small-set vertex expanders, graphs in which every sufficiently small subset of nodes has a large neighborhood, ensuring strong local connectivity.
In particular, when the expansion parameter satisfies $\phi > 1$, these graphs have a diameter of $o(\log n)$, as opposed to other standard notions of expansion.
Since the graph's diameter is a lower bound on the number of rounds required for rumor spreading, this makes small-set expanders particularly well-suited for fast information dissemination.
We prove that $k$-\texttt{PUSH\&PULL} takes $O(\log_{\phi} n \cdot \log_{k} n)$ rounds in these expanders, with high probability. 
We complement this with a simple lower bound of $\Omega(\log_{\phi} n+ \log_{k} n)$ rounds.


\end{abstract}
\vfill
\lowcotwo\ This is a low-co2 research paper: \lowcotwourl[\lowcotwoversion]. This research was developed, written, submitted and presented without the use of air travel.

\paragraph*{Funding.} This project has received funding from the European Research Council (ERC) under the European Union's Horizon 2020 research and innovation programme (grant agreement No 947702) and is supported by the Austrian Science Fund (FWF): P 32863-N and P 36280-N.

\thispagestyle{empty}

\newpage
\tableofcontents

\thispagestyle{empty}
\newpage

\end{titlepage}

\section{Introduction}
Rumor spreading, which is loosely inspired by biological and social phenomena, is one of the most well-studied stochastic processes on graphs with a rich history in refined analyses~\cite{FriezeG85,Pittel87,FeigePRU90,KarpSSV00,Mosk-AoyamaS08,FountoulakisP13,ChierichettiLP11,ChierichettiGLP18,DoerrFF11,SauerwaldS11,DoerrF11,DoerrFF12,FountoulakisPS12,GiakkoupisS12,AvinE18,Giakkoupis14,DoerrFS14,DoerrK14,GiakkoupisNW16,MehrabianP16,AcanCMW15,PanagiotouPS015,DaumKM20,GhaffariN16,DoerrK17}.
In the classic \textit{random phone call model}~\cite{DemersGHIL87}, an arbitrary node of the graph initially knows a piece of information, the rumor, which spreads across the graph in synchronous rounds until eventually all nodes are informed. 
In each round, every node calls one of its neighbors uniformly at random:
in the \texttt{PUSH} protocol each informed node informs the node it calls, in the \texttt{PULL} protocol each uninformed node is informed by the node it called, and in the \texttt{PUSH\&PULL} protocol both of these information exchanges happen simultaneously.

It is well known that \texttt{PUSH}, \texttt{PULL}, and \texttt{PUSH\&PULL} on a \textit{complete graph} with $ n $ nodes require $ \Theta (\log n) $ rounds to complete rumor spreading with high probability\footnote{We say that a statement holds \emph{with high probability} (w.h.p.\ for short) if it holds with probability at least $1-n^{-c}$, where $c$ is a positive constant hidden in the asymptotic notation.}~\cite{FriezeG85,KarpSSV00}.
Rumor spreading has been also extensively studied on \textit{expander graphs}, with known tight bounds that relate the completion time of the process with different notions of expansion.
In the following, we discuss only \texttt{PUSH\&PULL}, due to strong lower bounds that exist for both \texttt{PUSH} and \texttt{PULL} in isolation~\cite{ChierichettiGLP18,SauerwaldS11}.
In particular, the protocol with high probability requires $\Theta (\varphi^{-1} \log n)$ rounds for graphs with \textit{conductance} $\varphi \in [0,1]$~\cite{ChierichettiGLP18}, 
and $\Theta(\phi^{-1} \log n \log \Delta)$ rounds for graphs with maximum degree $\Delta$ and \textit{vertex-expansion} $\phi \in [0,1]$~\cite{GiakkoupisS12,Giakkoupis14}.
It is also known that \texttt{PUSH\&PULL} requires $ \Omega (\sqrt{n}) $ rounds with high probability in graphs with constant \textit{edge-expansion}~\cite{ChierichettiLP11}.
The study of rumor spreading protocols on expanders, subsequently, has been used as a building block for designing distributed information dissemination algorithms~\cite{HillelS10,Censor-HillelS12,Censor-HillelHKM17,Haeupler15}, some of which are based on expander decompositions~\cite{SpielmanT04}.

\textit{Faster rumor spreading} can be achieved on specific topologies or by modifying the process. 
On power-law degree distributed graphs, the process completes in $O(\log \log n)$ rounds if the exponent of the power-law $\beta$ is between $2$ and $3$, while it needs $\Omega(\log n)$ rounds if $\beta$ is greater than~$3$~\cite{FountoulakisPS12}.
A variant of the process, where nodes can use direct addressing, i.e., are able to contact neighbors whose address was learned before, takes $\Theta(\log \log n)$ rounds~\cite{haeupler2014optimal,AvinE18}.

With the goal of making rumor spreading even faster, in this paper we study a \textit{multi-call variant} of the classical protocol. We call this the \pushpull{$\PPPP$} protocol, where in each round every informed node samples $\PPPP$ neighbors and sends the rumor to them, while every uninformed node samples $\PPPP$ neighbors and requests the rumor from them.
The sampling happens uniformly at random and with replacement in both \texttt{PUSH} and \texttt{PULL} operations.

This process naturally interpolates between classical single-call rumor spreading and full broadcasting. It involves random selective communication, as in the former, while nodes simultaneously interact with multiple neighbors, as in the latter. 
This interpolation intuitively results in a trade-off between the amount of communication and the spreading time of the two processes. 
Classical rumor spreading requires minimal communication, while broadcasting is optimal in spreading time, matching the graph's diameter in the worst case.

Somewhat similar processes have been studied for the special case $k=4$ on random regular graphs~\cite{BerenbrinkEF16}, with a randomized number of simultaneous calls on the complete graph~\cite{PanagiotouPS15,DoerrK17}, and in the asynchronous setting considering multiple pulls on the complete graph~\cite{MocquardSA21,robin2022stochastic}.
For the related random-walk process, a multi-walk extension similar in spirit to our multi-call extension for rumor spreading has been studied by Alon et al.~\cite{AlonAKKLT11}.
Apart from these works we are not aware of any paper studying the multi-call setting explicitly.
While the analogy to a biological or social process might not be given anymore for the multi-call variant, we still believe that increasing the number of simultaneous calls is a quite natural generalization that captures a relevant aspect of information dissemination under bandwidth constraints.
Furthermore, note that the classic ``single-call'' variant allows each node to \emph{receive} multiple calls at the same time.\footnote{Restricted variants in which nodes receiving multiple calls can spread the rumor only to a single neighbor~\cite{DaumKM20} cannot achieve performance guarantees comparable to the unrestricted variant in conductance expanders and vertex expanders~\cite{GhaffariN16}.}
Therefore, it seems natural to also allow nodes to \emph{initiate} several calls at the same time.

This leads to the central question of our work:
\begin{mdframed}
    \centering\small
    \textbf{Question:} 
    \textit{How much communication is required for sub-logarithmic time rumor spreading?}
\end{mdframed}

\subsection{Our Contribution}
The starting point of our investigation is the study of multi-call rumor spreading on the \textit{complete graph}.
Using standard arguments that we present as a warm-up in \Cref{sec:complete_graphs}, we show that the \pushpull{$\PPPP$} protocol requires $ \Theta (\log_{\PPPP} n) $ rounds with high probability.
This gives an interesting trade-off between the classic single-call bound of $ \Theta (\log n) $ and a constant number of rounds for a polynomial number of simultaneous calls.

Next, we extend our question to \textit{expander graphs}.
Our goal is to generalize the trade-off that we get for \pushpull{$\PPPP$} on the complete graph to incorporate an expansion parameter $ \phi $, similarly to the works mentioned above for the single-call model.
While this is an intriguing combinatorial question in its own right, we believe that the recent use of expanders in graph neural networks~\cite{deac2022expander,ShirzadVVSS23} adds additional motivation for revisiting the foundations of information dissemination in expanders.

Our focus is on \textit{small-set expanders}, that have initially been studied for the notion of conductance (see e.g.,~\cite{RaghavendraS10,AroraBS15}), and later for edge expansion (see e.g.,~\cite{MoshkovitzS18}) and vertex expansion (see e.g., \cite{ChuzhoyMVZ16,LouisM16,ChlamtacDM17}). Small-set vertex expansion has also been identified as a crucial property of networks for designing routing protocols~\cite{AroraLM96}.
A non-negligible part of our contribution consists of studying properties of \textit{small-set vertex expanders} in a particular regime where the vertex expansion $\phi$ is larger than~$1$.
In particular, we establish upper and lower bounds on the \textit{diameter} of such graphs (\Cref{sec:expanders}), which, unlike other standard notions, is \textit{sub-logarithmic}.

This property is crucial for achieving sub-logarithmic rumor spreading time, as the graph diameter provides a fundamental lower bound on the number of rounds required.
Roughly speaking, our main result (see \Cref{sec:UB} and \Cref{thm:main_thm}) implies that, with high probability, a constant number of rounds suffices for multi-call rumor spreading in a wide regime of small-set vertex expanders when both the number of calls $k$ and the expansion parameter $ \phi $ are polynomial in $ n $.
We also give a complementary lower bound (see \cref{thm:LB}) showing that both parameters need to be polynomial to achieve a constant number of rounds.


\subsection{Formal Statement of Our Results}
\paragraph*{Small-set vertex expanders}
Let $G=(V,E)$ be an unweighted graph. 
For every $S \subseteq V$, we write $\overline S := V \setminus S$ for the complement of $S$. 
We also define the neighborhood of a set $S$ as $N(S):=  \{v\in V:\exists s\in S,\,\{s,v\}\in E\}$ and its boundary $\d S:=N(S)\cap \overline S$. 
For a single node $v$, we write $N(v):=N(\{v\})$. 
Note that $v\notin N(v)$, but $N(S)\cap S$ can be nonempty for larger $S$. 
Using this notation, we are ready to define the class of \textit{small-set vertex expanders} (see, e.g., \cite{ChuzhoyMVZ16,LouisM16,ChlamtacDM17}).  
\begin{definition}\label{def:vtx_exp}
    Let $\phi \in (0,n)$ and $\alpha\in (0,\tfrac{1}{2}]$. We say a graph $G$ with $ n $ nodes is a $(\phi,\alpha)$-\emph{(vertex) expander} if
    \begin{equation*}
        \min_{\substack{S \subseteq V \text{ s.t.}\\ 0<|S|\leq \alpha n}} \frac{|\d S|}{|S|} \geq \phi.
    \end{equation*}
\end{definition}
The name small-set comes from the fact that the expansion property holds for sets of size at most $\alpha n$, opposed to the classical definition where $\alpha=\frac{1}{2}$.
The regime $\alpha<\frac{1}{2}$ is conceptually different and has implications on the expansion parameter $\phi$, allowing $\phi > 1$.
When $\alpha=\tfrac{1}{2}$, in fact, we \emph{cannot} have $\phi > 1$.
Indeed, any set $S$ with $|S|=n/2$ nodes satisfies $\tfrac{|\d S|}{|S|}\le \tfrac{|V\setminus S|}{|S|}=1$. 
More generally, there are no $(\phi,\alpha)$-expanders for $\alpha> \tfrac{1}{1+\phi}$ (see \Cref{lm:Imp_Exp}). 

However, if we take $\alpha= \tfrac{1}{1+\phi}$, we only get a very restricted class of graphs: such $(\phi,\alpha)$-expanders have extremely low diameter and are very dense; they have diameter 2 and minimum degree $\Theta(n)$ (see \Cref{lm:Restricted_Exp}). 
On the other hand, if we take $\alpha$ too small, the expansion property becomes local and the graph is no longer necessarily connected. To be precise, for $\alpha\leq \tfrac{1}{2+2\phi}$, there exist $(\phi,\alpha)$-expanders that are disconnected (see \Cref{lm:Disc_Exp}). 
In between these two bounds, we get a guarantee on the diameter that depends on $\phi$. 
For $\tfrac{1}{2+2\phi}<\alpha \leq \tfrac{1}{1+\phi}$, we have that $(\phi,\alpha)$-expanders have diameter at most $O(\log_\phi n)$ (see \Cref{lm:Exp_Diam}). 
We note that for $\phi=\omega(1)$ the diameter is $o(\log n)$ and that for polynomial~$\phi$ it becomes constant. 

Next, we provide two examples of vertex expanders where the bound on the diameter is tight.
First of all, we obtain that complete graphs are $(\phi,\alpha)$-vertex expanders for $0<\phi \leq n-1$ and $\alpha \leq \tfrac{1}{1+\phi}$ (see \Cref{lm:Complete_Exp}).
Secondly, we show that random graphs are vertex expanders. More precisely, we show that Erdős-Rényi random graphs, where between each pair of nodes there is an edge with probability $p=\tfrac{3\phi}{n}$, are $(\phi,\alpha)$-vertex expanders for $\alpha\leq \tfrac{1}{1+1.6\phi}$ (see \Cref{lm:randomgraph}). 
We note that in this regime of $p$ these graphs have diameter $\Theta(\log_\phi n)$ with high probability~\cite{ChungL01}.

We also note that concurrent work by Hsieh et al.~\cite{hsieh2025explicit} gives an explicit $\Theta(\phi)$-regular construction for small-set vertex expanders.

\paragraph*{Rumor Spreading}


For $\tfrac{1}{2+2\phi}<\alpha \leq \tfrac{1}{1+\phi}$, we study rumor spreading on $(\phi,\alpha)$-vertex expanders through the \pushpull{$\PPPP$} protocol. 

\begin{restatable}{theorem}{mainthm}\label{thm:main_thm}
    Let $\phi>1$, $\alpha > \tfrac{1}{2+2\phi}$, let $G=(V,E)$ be a $(\phi,\alpha)$-expander, and let $\PPPP> \log^3 n$. Then w.h.p.\ rumor spreading with the \pushpull{$\PPPP$} protocol requires the following number of rounds:
    \[
        O\left(\left(\log_{\phi} n +\phi^{-1}\big(\alpha - \tfrac{1}{2+2\phi}\big)^{-1}\right)\log_{\PPPP} n\right).
    \] 
\end{restatable}
Before commenting on this round complexity, we give a lower bound
that follows immediately from the fact that on complete graphs we need $\Omega(\log_\PPPP n)$ rounds (see \Cref{lm:Kn_LB}) and on Erd\H{o}s-R\'enyi random graphs (which are $(\phi,\alpha)$-expanders by \Cref{lm:randomgraph}) we need $\Omega(\log_\phi n)$ rounds, since they have diameter $\Omega(\log_\phi n)$~\cite{ChungL01}.
\begin{restatable}{theorem}{thmLB}\label{thm:LB}
    Let $n\ge 1$, $\phi>1$ , $\alpha \leq \tfrac{1}{1+1.6\phi}$ and $k\ge 2$. Then there exist a $(\phi,\alpha)$-expander $G=(V,E)$ on $n$ nodes such that w.h.p.\ rumor spreading with the \pushpull{$\PPPP$} protocol requires $\Omega(\log_\phi n +\log_\PPPP n)$ rounds.
\end{restatable}
We make the following observations about the round complexity we get in \cref{thm:main_thm}:
\begin{itemize}
    \item The round complexity goes to infinity as $\alpha$ approaches $\tfrac{1}{2+2\phi}$, where expanders can be disconnected and hence rumor spreading never finishes. Interestingly, an analogous term appears in the analysis of the classic protocol on random graphs, making the rumor spreading time go to infinity in the non-connected regime~\cite{PanagiotouPS015}.
    \item For $\alpha\geq \tfrac{1}{2+2\phi}+\Omega\left(\tfrac{1}{\phi \log_\phi n}\right)$ the number of rounds simplifies to $O(\log_\PPPP n \log_\phi n)$.
    \item If additionally $\phi=n^{\Omega(1)}$, then we obtain $O(\log_\PPPP n)$, which is optimal due to \cref{thm:LB}. 
    \item If instead $\PPPP=n^{\Omega(1)}$, then we obtain $O(\log_\phi n)$, which is optimal due to \cref{thm:LB}.
\end{itemize}

We believe that our work opens several interesting research directions for the community. 
Closing the gap between our upper and lower bound is left as an open problem.
Extending the proof technique of the state of the art analysis \cite{Giakkoupis14} to the multi-call setting in our opinion appears to be non trivial; it is also not clear if this would be sufficient to actually close the gap. 
Another direction is to explore the use of such expanders in applications beyond rumor spreading. They can serve as a natural alternative to traditionally used low-diameter graphs, with the benefit of having combinatorial expansion properties.

\paragraph*{Lower bounds}
The fact that small-set vertex expanders allow $\phi>1$ and hence sub-logarithmic diameter is fundamental to achieve our results.
We observe that the classical notions of expansions considered in the literature do not have these properties and, therefore, the multi-call variant cannot speed up the rumor spreading time on these graphs.
The most studied notions of expanders for which these properties do not hold are the following:
\begin{enumerate}
    \item \emph{$\phi$-conductance expanders}, where for $\phi \in (0,1)$ we have 
    \[
        \min_{\substack{S \subseteq V \text{ s.t.}\\ 0<\Vol(S)\leq m/2}} \tfrac{|E(S,\overline{S})|}{\Vol(S)}\geq \phi,
    \]
    where $|E(S,\overline{S})|$ is the number of edges crossing the cut $(S,\overline{S})$, and $\Vol(S) := \sum_{u \in S} \deg(u)$ is the volume of $S$, with $\deg(u)$ being the degree of a node $u \in V$.
    \item \emph{$\phi$-edge expanders}, where for $\phi \in (0,n)$ we have
    \[
        \min_{\substack{S \subseteq V \text{ s.t.}\\ 0<|S|\leq n/2}} \tfrac{|E(S,\overline{S})|}{|S|}\geq \phi.
    \]
    \item \emph{$\phi$-vertex expanders}, where for $\phi \in (0,1)$ we have
    \[
        \min_{\substack{S \subseteq V \text{ s.t.}\\ 0<|S|\leq n/2}} \tfrac{|\d S|}{|S|}\geq \phi.
    \]
\end{enumerate}

Regarding conductance-expanders, the following statement holds.
\begin{lemma}
    Let $n\ge 1$, $\phi \in (0,1)$ and $k\ge 1$. There exists a $\phi$-conductance expander on $n$ nodes such that \pushpull{$k$} takes $ \Theta (\phi^{-1} \log n) $ rounds.
\end{lemma}
The proof follows since the same upper bound holds already for $k=1$, and the lower bound comes from the diameter of $\phi$-conductance expanders~\cite{ChierichettiGLP18}.

Rumor spreading on $\phi$-edge expanders on the other hand, is not very appealing in the first place due to the strong lower bound of $ \Omega (\sqrt{n}) $ rounds for the single-call variant -- which even holds in edge expanders with low diameter~\cite{ChierichettiLP11}.

Regarding vertex expanders, we have the following statement.
\begin{lemma}
    Let $n\ge 1$, $\phi \in (0,1)$ and $k\ge 1$. There exists a $\phi$-vertex expander on $n$ nodes such that \pushpull{$k$} takes $ \Omega (\phi^{-1} \log n) $ rounds.
\end{lemma}
The lemma follows from a lower bound on the diameter of $\phi$-vertex expanders~\cite{GiakkoupisS12}.
A tight bound of $\Theta(\phi^{-1} \log n \log \Delta)$ is known for the single-call variant when $ \phi \leq 1 $~\cite{GiakkoupisS12,Giakkoupis14}. 
This means that the potential for savings in the multi-call variant is limited to a single $ \log \Delta $-factor compared to the single-call variant; eliminating this factor is, to the best of our judgment, a quite fine-grained endeavor that we deliberately omit from this paper with its more conceptual focus.

Another class of graphs with very good expansion properties is that of \textit{Ramanujan graphs}, whose spectral gap in the matrix representation of the graph is almost as large as possible.
In particular, they can achieve sublogarithmic diameter, but only under very restrictive conditions, i.e., when they are $\Theta(n)$-regular~\cite{sardari2019diameter}.
In contrast, our class of small-set vertex expanders allows sublogarithmic diameter while being significantly sparser.


\subsection{Technical Challenges}
The proof we provide for rumor spreading in vertex expanders with $\phi> 1$ is based on the ideas for the case of expansion at most 1, by Giakkoupis and Sauerwald~\cite{GiakkoupisS12}. There are three complications we need to overcome. Let $I_t$ denote the informed nodes in round $t$. We will show that in any round either $I_t$ grows significantly or the boundary $\d I_t$ grows significantly. While this growth is a constant factor for the case $\phi<1$, leading to a $\log n$ in the round complexity, we need to grow more aggressively. Intuitively, if the expected growth in one round is $\mu$, then the expected growth in a round with $\PPPP$-parallel calls is $\PPPP \mu$. However, this brings us to the first issue. 

\paragraph*{Overlap of Parallel Calls}
To some extent this challenge already appears in the \texttt{PUSH} model: the fact that two nodes individually push the rumor to a neighbor does not mean two new nodes are informed; they can push the rumor to the same node. 
This phenomenon is exacerbated when we call $\PPPP$ times. 
Moreover, it now also impacts pulling: if one node has two successful parallel pulls in a round, then this does not contribute as two nodes pulling. In other words, we cannot simply sum the expected gain from separate pushes and pulls to obtain the expected number of newly informed nodes. Let us consider these probabilities to understand this situation better. 

Let $u\in \d I_t$ be an uninformed node. The probability that $u$ pulls the rumor from $I_t$ in one round of \pushpull{$\PPPP$} is
\[
    1-\left(1-\frac{|N(u) \cap I_t|}{\deg(u)}\right)^\PPPP = p.
\]
Now if $\deg(u)\geq \PPPP |N(u) \cap I_t|$, we can use a binomial expansion to lower bound this probability by $\frac{\PPPP |N(u) \cap I_t|}{2\deg(u)}$, which is a factor $\PPPP/2$ higher than the probability that a single call succeeds. For nodes with lower degree we note that $p \geq 1/2$, so we can still say that they get informed with good probability. However, this does not carry over a factor $\PPPP$ compared to the single-call situation. 
We are able to show a sufficient growth by carefully analyzing the number of nodes with high and low degrees and using different arguments for each case. 

In this brief description we have only highlighted the issue. 
The actual solution is more involved but based on this intuition. 
For details see \Cref{sec:medium} and \Cref{sec:high}.

\paragraph*{Probabilistic Guarantees}
The second issue we need to deal with is the probability of successfully growing $I_t$ or $\d I_t$ through \pushpull{\PPPP}. 
Intuitively, there are two sides to what is going on here: on the negative side, we have only a sub-logarithmic number of rounds, which is often not great to give bounds with high probability. On the positive side, we have significant growth in each step, which is good for tail bounds. 

From this second observation, we see that in some cases a Chernoff bound suffices. However, in some situations, we cannot use a simple Chernoff bound since the random variables are neither independent nor negatively correlated. In particular, this happens when we want to show that the boundary $\d I_t$ grows. 
Let $u \in \d I_t$ be a node in the boundary of $I_t$ and $v_1,v_2 \in N(u) \setminus I_t$ be two uninformed neighbors of $u$.
The events ``$v_1 \in \d I_{t+1}$'' and ``$v_2 \in \d I_{t+1}$'' are not independent nor negatively associated as both may occur as a consequence of the event ``$u \in I_{t+1}$''. 
In some cases we can use the bounded difference inequality, see \Cref{sc:tail_bounds}, that can give tail bounds on \emph{functions} of independent random variables which we design to be equal to the sum of the correlated ones discussed above. 
See for example \Cref{sec:medium} and \Cref{sec:high}. 

In other cases, the function that measures progress does not satisfy the conditions for the bounded difference inequality. Here we use an additional trick: we do not use the full potential of $\PPPP$ parallel calls, but only $\PPPP/\log n$ calls. Each batch of $\PPPP/\log n$ call succeeds with constant probability, and since we have $\log n$ of these independent batches in parallel, at least one succeeds with high probability. This comes at the drawback that we only grow a factor $\PPPP/\log n$. By our assumption that $\PPPP > \log^3 n$, we get $\PPPP/\log n> \PPPP^{2/3}$, which means that $\log_{\PPPP/\log n}n=O(\log_\PPPP n)$.

\paragraph*{No Expansion for Large Sets}
The last complication we point out here is that small-set expansion is, by definition, only guaranteed for sets $S$ of size at most $\alpha n$. 
Therefore, we cannot use the expansion property to argue that $I_t$ keeps growing until it hits $n/2$ nodes---at which point a standard symmetry argument shows that the process completes within a factor 2 in the round complexity. 
To handle the case where more than $\alpha n$ nodes are informed, we show that the size of the boundary $\d I_t$ is at least a constant fraction of $n$, meaning that if the set of informed nodes $I_t$ keeps growing by a constant fraction of the boundary $\d I_t$, the process still completes in a constant number of additional iterations. 
To be precise, this leads to the factor $\phi^{-1}(\alpha - \tfrac{1}{2+2\phi})^{-1}$ in our round complexity. 

 \newpage
\section{Warm-Up: Multi-Call Rumor Spreading on Complete Graphs}\label{sec:complete_graphs}
This section serves as a warm-up for our main result; we investigate what impact parallel calls have on rumor spreading on complete graphs. We provide a full characterization with matching upper and lower bounds. For $k=1$, we know that rumor spreading takes $ \Theta (\log n) $ rounds~\cite{FriezeG85,KarpSSV00}, so we only consider $k\geq 2$.
The analysis in \cite{DoerrK17} yields precise results that include the leading constant (in fact tight apart from an additive number of rounds) for the case $k = O(1)$. However, it does not extend to super-constant values of $k$, where it leads to an upper bound of $O(k \cdot \log_k n)$, which is asymptotically worse than ours.

The proofs in the section are generalizations of standard arguments for rumor spreading on cliques. 
Independent work by Out, Rivera, Sauerwald, and Sylvester~\cite{OutRS024} considers rumor spreading with a time-dependent credibility function. Some of their proofs resemble the results in this section. 

In the proof for the upper bound, we use the following standard lemma, which is a trivial generalization of the version with $\PPPP=1$ (see, e.g., \cite[Lemma 3]{ChierichettiLP10}, \cite[Lemma 3.3]{GiakkoupisS12}, or \cite[Lemma 2.1]{Censor-HillelHKM17}). 

\begin{lemma}\label{lm:symmetry}
    For $S,T\subseteq V$, let $T_{\pushpull{$\PPPP$}}(S,T)$ be the number of rounds for \pushpull{$\PPPP$} until a rumor that is initially known to all nodes in $S$ to spread to at least one node of $v\in T$. Let $V_1,V_2\subseteq V$. Then the random variables $T_{\pushpull{$\PPPP$}}(V_1,V_2)$ and $T_{\pushpull{$\PPPP$}}(V_2,V_1)$ have the same distribution.
\end{lemma}

\begin{lemma}
   Rumor spreading in a complete graph with the \pushpull{$\PPPP$} protocol for every $k\ge 2$ requires $O(\log_\PPPP n)$ rounds with high probability. 
\end{lemma}
\begin{proof}
An upper bound of $O(\log n)$ has been known for a long time for $k=1$~\cite{FriezeG85,KarpSSV00}. Since parallel spreading can only make the process faster, we have this upper-bound for any $\PPPP$. In particular, for any constant $k$ we have that the \pushpull{$\PPPP$} protocol requires $O(\log n)= O(\log_\PPPP n)$ rounds with high probability. In the following we assume $k\geq 17$. 

Let $I_t$ denote the informed nodes at round $t$. 
First we show using only $k$-\texttt{PUSH} that $|I_t|\geq (\PPPP/4)^t$ until $\PPPP |I_t|> n-1$. 
Assume $\PPPP |I_t|\leq n-1$. Each node receives a message from at least one node in $I_t$ with probability $1-(1-\tfrac{1}{n-1})^{\PPPP |I_t|}=\tfrac{\PPPP|I_t|}{n-1}-\tfrac{1}{2}\PPPP|I_t|(\PPPP|I_t|-1)\tfrac{1}{(n-1)^2}+\dots \geq \tfrac{\PPPP|I_t|}{2(n-1)}$ where we use binomial expansion.
Now we see that 
\begin{align*}
    \E[|I_{t+1}\setminus I_t|] &\geq  (n-|I_t|) \frac{\PPPP|I_t|}{2(n-1)}
    \geq \left(n-\frac{n-1}{\PPPP}\right)\frac{\PPPP|I_t|}{2(n-1)}
    \geq \tfrac{\PPPP-1}{2}|I_t|.
\end{align*}
$|I_{t+1}\setminus I_t|$ can be written as a sum of 0/1 random variables. 
Note that they are not independent, but they are negatively associated since the probability that $u$ receives a push can only decrease under the assumption that $v$ receives a push. 
Hence, as discussed in the appendix (see \cref{thm:chernoff bound}), Chernoff still gives
\begin{align*}
    \P[|I_{t+1}\setminus I_t|< \tfrac{\PPPP-1}{4}|I_t|] \leq e^{-(\PPPP-1)|I_t|/16}.
\end{align*}
For $\PPPP\geq 16(c+1)\log n$, this is bounded by $n^{-(c+1)}$, which gives the required probability. For smaller $k$, we need to analyze the process more carefully. 

We consider $\beta\cdot 2\log_k n$ rounds, for some constant $\beta=\beta(c)\geq 2$ to be decided later. By a Chernoff bound, the probability that any such round does \emph{not} have the required growth is bounded by $e^{-(k-1)/16}$. We compute the probability that more than $(\beta-1)\cdot 2\log_k n$ rounds fail in reaching the required growth. 
In fact, if at most $(\beta-1)\cdot 2\log_k n$ rounds fail, we have at least $2\log_k n \geq \tfrac{\log(n/k)}{\log(k/4)}$ successes which is enough to reach $ |I_t|> (n-1)/k$ nodes. 

To count the number of failures, we see that this is dominated by a binomial random variable $B(N,p)$ with $N=\beta\cdot 2\log_k n$ trials and probability of success $p=e^{-(k-1)/16}$. We bound the probability that $B(N,p)\geq (\beta-1)\cdot 2\log_k n= \tfrac{\beta-1}{\beta}N$.
By using standard tail bounds on the binomial distribution (see \Cref{thm:binom_tail}) we get that for $p<\tfrac{\beta-1}{\beta}$ we have
\begin{equation*}
    \P\left[ B(N,p) \geq \frac{\beta-1}{\beta}N\right] \leq \exp\left( -N \left(\frac{\beta-1}{\beta}\log\left( \frac{\frac{\beta-1}{\beta}}{p}\right)+\left(1-\frac{\beta-1}{\beta}\right)\log\left(\frac{1-\frac{\beta-1}{\beta}}{1-p}\right) \right) \right).
\end{equation*}
Note that we have $p<\tfrac{\beta-1}{\beta}$, since $p=e^{-(k-1)/16}\leq e^{-1} < 1/2 \leq \tfrac{\beta-1}{\beta}$, where the last inequality follows from the assumption that $\beta\geq 2$. 

Now we simplify this bound 
{\small%
\begin{align*}
    &\P\left[ B(N,p) \geq \frac{\beta-1}{\beta}N\right] \leq \exp\left( -N \left(\frac{\beta-1}{\beta}\log\left( \frac{\frac{\beta-1}{\beta}}{p}\right)+\left(1-\frac{\beta-1}{\beta}\right)\log\left(\frac{1-\frac{\beta-1}{\beta}}{1-p}\right) \right) \right)
    \\
    &\qquad\qquad= \exp\left( -\beta 2\log_k (n) \left(\frac{\beta-1}{\beta}\log\left( \frac{\frac{\beta-1}{\beta}}{e^{-(k-1)/16}}\right)+\frac{1}{\beta}\log\left(\frac{\frac{1}{\beta}}{1-e^{-(k-1)/16}}\right) \right) \right)\\
    &\qquad\qquad\leq \exp\left( - 2\log_k (n) \left((\beta-1)\left( \log\left( \frac{\beta-1}{\beta}\right)+(k-1)/16\right)-\log\beta \right) \right). 
\end{align*}}
We add the constraint that $\beta\geq 3$, so that $\log\left(\tfrac{\beta-1}{\beta}\right) > -1/2 \geq -(k-1)/32$, hence we have
\begin{align*}
    \P\left[ B(N,p) \geq \frac{\beta-1}{\beta}N\right] 
    &\leq \exp\left( - 2\log_k (n) \cdot (\beta-1)(k-1)/32 -\log\beta\right)
    \\
    &= \exp\left( - 2\log_k (n) \cdot (\beta-1)\left( (k-1)/32-\frac{\log\beta}{\beta-1}\right) \right).
\end{align*}
We add the further constraint that $\beta\geq 11$, so that $\tfrac{\log\beta}{\beta-1}<1/4 \leq (k-1)/64$. This gives 
\begin{align*}
    \P\left[ B(N,p) \geq \frac{\beta-1}{\beta}N\right] 
    &\leq \exp\left( - 2\log_k (n)\cdot (\beta-1) (k-1)/64 \right)
    \\
    &= \exp\left( - \log (n)\cdot \frac{\beta-1}{32}\cdot \frac{k-1}{\log k} \right)
    \leq n^{-(c+1)},
\end{align*}
where the last inequality holds for $\beta \geq 32c+33$. 
In other words, we have shown that with high probability after at most $O(\log_\PPPP n)$ rounds we have that $\PPPP |I_t|> n-1$. 

Next, we show that $|I_{t+1}|> n/2$ w.h.p., so assume $|I_t|\leq n/2$. Let $v\in V\setminus I_t$. Then the probability that $v$ pulls from $I_t$ is 
\begin{align*}
    1-\big(1-\tfrac{|I_t|}{n-1}\big)^\PPPP \geq  1-(1-\tfrac{1}{\PPPP})^\PPPP \geq 1-1/e.
\end{align*}
So the expected number of nodes that pull is at least $(n-|I_t|)(1-1/e)$. 
In this case, $|I_{t+1}\setminus I_t|$ is a standard binomial random variable (i.e., sum of independent 0/1 random variables). Hence, using Chernoff we see that 
\begin{align*}
    \P[|I_{t+1}\setminus I_t|\leq \tfrac{n-|I_t|}{2}] &= \P[|I_{t+1}\setminus I_t|\leq (1-\tfrac{1+1/e}{2})(n-|I_t|)(1-1/e)] \\
    &\leq e^{-\left(\tfrac{1+1/e}{2}\right)^2(n-|I_t|)(1-1/e)/2 }
    \leq e^{-n/14},
\end{align*}
using that $n-|I_t|\geq n/2$. Now we have that $|I_{t+1}|> |I_t|+ \tfrac{n-|I_t|}{2}\geq n/2$. 

By a standard symmetry argument, see \Cref{lm:symmetry}, we see that we use at most $2O(\log_\PPPP n)=O(\log_\PPPP n)$ rounds with probability $1-n^{-c}$.     
\end{proof}

Next, we show the matching lower bound. 

\begin{lemma}\label{lm:Kn_LB}
    Rumor spreading in a complete graph with the \pushpull{$\PPPP$} protocol for every $k\ge 2$ requires $\Omega(\log_\PPPP n)$ rounds with high probability.
\end{lemma}
\begin{proof}
    We show that it takes $\Omega(\log_\PPPP n)$ rounds to get $n/\PPPP$ nodes informed. 
    First, let us assume that we are at a round $t_0$ such that $54(c+1) \log n \le|I_{t_0}| \le 216(c+1)\, k \log n$ for some positive constant $c$. 
    We will later show that this round exists w.h.p.
    With such an assumption, we show by induction that $|I_{t_0+t}|\leq (3\PPPP)^t |I_{t_0}|$, until $|I_t|\geq n/k$.\footnote{Note that this requires that $216(c+1)\, k\log n< n/k \iff k< \sqrt{n/(216(c+1)\log n)} $, which we can assume since $1=\Omega(\log_k n)$ is a trivial lower bound for such $k$.} 
    Clearly $k$-\texttt{PUSH} can inform at most $\PPPP |I_t|$ nodes in one round. 
    So it remains to show that the new informed nodes due to $k$-\texttt{PULL} in one round are at most $2\PPPP |I_t|$ w.h.p.

    Let $v\in V\setminus I_t$. Then the probability that $v$ pulls from $I_t$ equals
    \begin{align*}
        1-\left(1-\frac{|I_t|}{n-1}\right)^\PPPP \leq \frac{3}{2} \PPPP\frac{|I_t|}{n-1}.
    \end{align*}
    So we can bound the expected gain from pull $|I^{\rm{(pull)}}_{t+1}|$ by
    \begin{align*}
        \E[|I^{\rm{(pull)}}_{t+1}|]\leq (n-|I_t|)\cdot \frac{3}{2} \PPPP\frac{|I_t|}{n-1} \leq \frac{3}{2} \PPPP |I_t|.
    \end{align*}
    Note that $|I^{\rm{(pull)}}_{t+1}|$ is a binomial random variable, hence we can use Chernoff to bound the probability that this is bigger than $2k|I_t|$ and get
    \begin{align*}
        \P[|I^{\rm{(pull)}}_{t+1}|>2\PPPP|I_t|]
        = \P\left[|I^{\rm{(pull)}}_{t+1}|>\left(1+\frac{1}{3}\right)\frac{3}{2}\PPPP|I_t|\right]
        \leq e^{-\E[|I^{\rm{(pull)}}_{t+1}|]/27}.
    \end{align*}
    To see that this probability is indeed low, we lower bound the expected increase $|I^{\rm{(pull)}}_{t+1}|$ as well. 
    The probability that any $v\in V\setminus I_t$ pulls from $I_t$ equals
    \begin{align*}
        1-\left(1-\frac{|I_t|}{n-1}\right)^\PPPP \geq  \PPPP\frac{|I_t|}{2(n-1)}, 
    \end{align*}
    since we are assuming $|I_t|\leq n/\PPPP$. So we see
     \begin{align*}
        \E[|I^{\rm{(pull)}}_{t+1}|]
        \geq (n-|I_t|)\PPPP\frac{|I_t|}{2(n-1)} 
        \ge \frac{n-n/k}{2(n-1)}k |I_t|
        = \frac{n}{2(n-1)}(k-1)|I_t| 
        \ge \frac{k-1}{2}|I_t|
        \ge \frac{|I_t|}{2}.
    \end{align*}
Since we have by assumption that $|I_t|\geq|I_{t_0}|\ge 54(c+1)\log n$ the result holds at each step with probability at least $1-n^{-(c+1)}$. 
    If $|I_{t_0+t}|\leq (3k)^t |I_{t_0}|$, then we need 
    \begin{equation*}
        \frac{\log\left(\frac{n/k}{|I_{t_0}|}\right)}{\log(3k)} = \Omega(\log_k n) 
    \end{equation*}
    rounds to inform at least $n/k$ nodes. 

    It remains to show that until we inform $54(c+1) \log n$ nodes, we cannot suddenly overshoot $216(c+1)\, k \log n$ nodes. We note that $I_{t+1}=I_t\cup I^{\rm{(push)}}_{t+1}\cup I^{\rm{(pull)}}_{t+1}$. It is clear that $k$-\texttt{PUSH} can inform at most $k\cdot 54(c+1) \log n$ nodes. So it remains to show that it is very unlikely that the process reaches more than $216(c+1)\, k \log n$ nodes through $k$-\texttt{PULL}. 
    We do this by applying a Chernoff bound in a different way. More precisely, assume that $|I_t|\leq 54(c+1)\log n$, then we show that the probability that $|I^{\rm{(pull)}}_{t+1}|> 108(c+1)\, k\log n$ is small:
    {\small%
    \begin{align*}
        &\P[|I^{\rm{(pull)}}_{t+1}|>108(c+1)\, k\log n] =  \P\left[|I^{\rm{(pull)}}_{t+1}|>\left(1+\left(\frac{1}{2}+\frac{3}{2}\right)\frac{(54(c+1)\, k\log n}{\E[|I^{\rm{(pull)}}_{t+1}|]}-1\right)\E[|I^{\rm{(pull)}}_{t+1}|]\right] \\
        &\quad\leq \P\left[|I^{\rm{(pull)}}_{t+1}|>\left(1+\frac{ 27(c+1)\, k\log n}{\E[|I^{\rm{(pull)}}_{t+1}|]}\right)\E[|I^{\rm{(pull)}}_{t+1}|]\right]
        \leq \exp\left(-\frac{\left( \frac{ 27(c+1)\, k\log n}{\E[|I^{\rm{(pull)}}_{t+1}|]} \right)^2}{2+\frac{ 27(c+1)\, k\log n}{\E[|I^{\rm{(pull)}}_{t+1}|]}} \cdot \E[|I^{\rm{(pull)}}_{t+1}|] \right)
        \\
        &\quad\leq \exp\left(- (27/7)(c+1)\, k\log n\right)
        \le n^{-(c+1)}.
        \qedhere
    \end{align*}}
\end{proof}
 \newpage
\section{Multi-Call Rumor Spreading on Small-Set Vertex Expanders}\label{sec:UB}

The goal of this section is to prove the following theorem, which extends the study of multi-call rumor spreading on small-set vertex expanders with expansion $\phi>1$ (see \Cref{def:vtx_exp}).

\mainthm*

The proof of this theorem is based on the ideas for the case of expansion at most 1, by Giakkoupis and Sauerwald~\cite{GiakkoupisS12}. Where applicable, we follow their notation. 

Let $I_t$ denote the set of informed nodes in round~$t$. Our intermediate goal is to show that either $I_t$ grows or $\d I_t$ grows. We do this by partitioning the nodes in the boundary according to their degree, and then analyzing each set separately. More formally, we partition the nodes in the boundary $\d I_t$ into different sets $A_i$ defined as
\[
    A_i := \{ u \in \d I_t : d_i \le \deg(u) < 2 d_i \},
    \text{ with }
    d_i := 2^{i-1}.
\]
Note that there are at most $\log n$ of such sets.
We consider only those $A_i$'s that are sufficiently big and such that the degree of the nodes in the set is bounded by quantities depending on the size of the set itself or is large with respect to the size of the boundary $\d I_t$. 
Formally, let
\begin{equation}\label{eq: A_i good index}
    \mathcal{I} := \left\{
        i : |A_i| \ge \frac{|\d I_t|}{4\log n}
        \land \big(
            d_i \le 16\,|A_i| 
            \lor d_i \ge 2\,|\d I_t|
        \big)
    \right\}
\end{equation}
be the set of indices of the $A_i$'s we take into account.
We note that by considering such sets only, we consider at least half of the nodes in the boundary $\d I_t$. We include a proof for completeness. 
\begin{lemma}[\cite{GiakkoupisS12}]\label{lem: Ai's half the boundary}
It holds that
\( 
    \left|\bigcup_{i \in \mathcal{I}} A_i \right| 
    \ge |\d I_t| / 2
\).
\end{lemma}
\begin{proof}
The worst case with respect to the condition $|A_i| \ge \frac{|\d I_t|}{4\log n}$ is that all but one of the $A_i$'s have size $|A_i| = \lceil |\d I_t| / 4 \log n \rceil -1 < |\d I_t| / 4 \log n$.
The worst case with respect to the condition $d_i \le 16\,|A_i| \lor d_i \ge 2\,|\d I_t|$ is that $|A_i|=\lceil d_i/16\rceil-1\le d_i/16$ and $d_i< 2|\d I_t|$, for all $i < \log(2|\d I_t|)$. 
Hence, since the $A_i$'s are disjoint, we have
\[
    \left| \bigcup_{i \not\in \mathcal{I}} A_i \right|
    = \sum_{i \not\in \mathcal{I}} \left| A_i \right|
    \le \frac{|\d I_t|}{4\log n} \log n + \frac{1}{16} \left( 2|\d I_t| + |\d I_t| + |\d I_t| / 2 + |\d I_t| / 4 + \ldots + 2 \right)
    \le \frac{|\d I_t|}{2}
\]
which implies the thesis.
\end{proof}

In \cref{sec:mainlemma} we show the following lemma, which is the main technical lemma used in the proof of \cref{thm:main_thm}. 

\begin{restatable}{lemma}{supportlemma}\label{lem:support lemma Ai}
    For each $i\in \mathcal I$ and at every round $t$, 
    we have that with high probability 
    either at least $\frac{|A_i|}{128}$ nodes from $A_i$ are informed within $r=O(\log_\PPPP n)$ rounds, 
    or the boundary $\d I_t$ grows by at least $\frac{\PPPP^{1/6}}{64}|\d I_t|$ within $O(1)$ rounds.
\end{restatable}

\cref{lem: Ai's half the boundary,lem:support lemma Ai} lead to the following claim. 
\begin{lemma}\label{lm:main_lemma}
If $|I_t|\leq \tfrac{n}{2}$, for $r=O(\log_\PPPP n)$ we have that with high probability 
\begin{equation*}
    |I_{t+r}\setminus I_t|\geq \frac{1}{256}|\d I_t|.
\end{equation*} 
\end{lemma}
\begin{proof}
    First, we show that in each phase of $r'$ rounds with high probability we are in either of the following cases:
    \begin{enumerate}[a)]
        \item $|I_{t+r'}\setminus I_t|\geq \frac{1}{256}|\d I_t|$, for $r'=O(\log_\PPPP n)$,\label{case:a}
        \item $|\d I_{t+r'}\setminus \d I_t|\geq \frac{\PPPP^{1/6}}{64} |\d I_t|$, for $r'=O(1)$.\label{case:b}
    \end{enumerate}

    \cref{lem:support lemma Ai} guarantees that for each $i\in \mathcal I$ with high probability either $\frac{|A_i|}{128}$ nodes from $A_i$ are informed within $r=O(\log_\PPPP n)$ rounds, or the boundary $\d I_t$ grows by at least $\frac{\PPPP^{1/6}}{64}|\d I_t|$ within $r=O(1)$ rounds.

    The thesis follows by considering two cases.
    If there is $i \in \mathcal{I}$ such that the boundary grows then the thesis follows immediately.
    Otherwise, it must hold that for all $i \in \mathcal{I}$ there are at least $\frac{|A_i|}{128}$ nodes in $A_i$ which get informed. 
    Since the $A_i$'s are disjoint, by \cref{lem: Ai's half the boundary} we conclude that
    \[
        |I_{t+r} \setminus I_t| \ge \sum_{i \in \mathcal{I}} \frac{|A_i|}{128} 
        = \frac{1}{128} \left|\bigcup_{i \in \mathcal{I}} A_i \right| 
        \ge \frac{|\d I_t|}{256}.
    \]
    Next, we notice that Case \ref{case:b} can only occur $O(\log_{\tfrac{\PPPP^{1/6}}{64}}n)=O(\log_{\PPPP}n)$ consecutive times before we cover the entire graph. Hence after $r=O(\log_{\PPPP}n)+O(\log_{\PPPP}n)=O(\log_{\PPPP}n)$ rounds, Case \ref{case:a} holds. 
\end{proof}

Next, we use this lemma to show that we reach more than half of the graph. 
We show that in $O(\log_\phi n)$ phases of $O(\log_\PPPP n)$ rounds we reach $\alpha n$ nodes.
Let $I_{(t)}$ denote the set of informed nodes in phase $t$.\footnote{Note that we use $I_t$ for the set of informed nodes in \emph{round} $t$.} 
We start by showing by induction that until $|I_{(t)}|\geq \alpha n$, we have $|\d I_{(t)}|\geq (\tfrac{\phi}{256})^t$. 

\begin{itemize}
\item \textbf{Base case.} By expansion, every node $v
\in V$ has $\deg(v)=|\d\{v\}|\geq \phi |\{v\}|=\phi$. So we have that $|\d I_{(1)}|\geq \phi$. 

\item \textbf{Induction step.} Suppose that $|\d I_{(t)}|\geq (\tfrac{\phi}{256})^t$. 
By \Cref{lm:main_lemma} have  $|I_{(t+1)}|\geq \tfrac{1}{256}|\d I_{(t)}|$.
So suppose $|I_{(t+1)}|\leq \alpha n$, then we have $|\d I_{(t+1)}|\geq \phi|I_{(t+1)}|$ by expansion of $I_{(t+1)}$. We obtain $|\d I_{(t+1)}|\geq \tfrac{\phi}{256}|\d I_{(t)}|\geq \tfrac{\phi}{256}(\tfrac{\phi}{256})^t\geq (\tfrac{\phi}{256})^{t+1}$, using the induction hypothesis. 
\end{itemize}
So we conclude that in $O(\log_{{\phi/256}}(\alpha n))=O(\log_\phi n)$ phases of $O(\log_\PPPP n)$ rounds we cover $\alpha n$ nodes, using in total $O(\log_\phi n \cdot \log_\PPPP n)$ rounds.

Now suppose that $\alpha n<|I_{(t)}|\leq \tfrac{n}{2}$. Let $S\subseteq I_{(t)}$ be any subset of size exactly $|S|=\alpha n$, so $|I_{(t)}\setminus S| \leq \tfrac{n}{2}-\alpha n.$ By expansion of $S$, we have $|\d S|\geq \phi |S|=\phi \alpha n$. We now see that 
\begin{equation*}
    |\d I_{(t)}| \geq |\d S \setminus I_{(t)}| = |\d S \setminus (I_{(t)}\setminus S)|\geq |\d S| - |I_{(t)}\setminus S|\geq  \phi \alpha n - (\tfrac{n}{2}-\alpha n) = (\phi+1)\alpha n- n/2.
\end{equation*}
By \Cref{lm:main_lemma} we now have to have that in each phase of $O(\log_{\PPPP} n)$ rounds, $I_{(t)}$ grows by $$\frac{1}{256}|\d I_{(t)}|\geq \frac{1}{256}((\phi+1)\alpha n- n/2).$$ 
So we cover more than $n/2$ nodes in the following number of phases
\begin{equation*}
    \frac{n/2}{\frac{1}{256}((\phi+1)\alpha n- n/2)} = O\left(\frac{1}{(\phi+1)\alpha - 1/2}\right) = O\left(\frac{1}{\phi(\alpha-\frac{1}{2+2\phi})}\right).
\end{equation*}

As each phase last $O(\log_\PPPP n)$ rounds, in total, we see that the rumor spreads to more than $n/2$ nodes in the following number of rounds
\begin{equation*}
    R=O\left(\left(\log_{\phi} n +\frac{1}{\phi\big(\alpha - \tfrac{1}{2+2\phi}\big)}\right)\log_{\PPPP} n\right).
\end{equation*}

Now \Cref{lm:symmetry} states that the number of rounds we need to spread the rumor from $I_t$ to a node $u\in V\setminus I_t$ is the same as the number of rounds we would need to spread the rumor from $u$ to anywhere in $I_t$. 
Since we have shown that we reach more than half of the graph in $R$ rounds, and that $|I_t|> n/2$, a rumor started at $u$ reaches $I_t$ in another $R$ rounds. 
We conclude that the spreading process finishes after $ 2R $ rounds in total, proving \cref{thm:main_thm}.

\subsection{Proof of \texorpdfstring{\Cref{lem:support lemma Ai}}{Main Lemma}}\label{sec:mainlemma}
In this section we prove our main lemma. 
\supportlemma*

We look at each set $A_i$ for $i \in \mathcal{I}$ separately, and we upper bound the number of rounds until
either a large fraction of nodes in $A_i$ gets informed
or the size of the boundary increases by a large factor as a result of nodes in $A_i$ being informed.
We distinguish three cases informally involving sets of nodes with low, medium, and high degree.
We define these cases according to the following conditions:
\begin{enumerate}
    \item Low degree: $d_i \le 1024 \cdot 96 \cdot \PPPP$.
    \item Medium degree: $1024 \cdot 96 \cdot \PPPP < d_i \le 16\,|A_i|$.
    \item High degree: $d_i \ge 2\,|\d I_t|$.
\end{enumerate}
Note that the three cases do not cover all nodes in the boundary $\d I_t$. However, they do cover all nodes in the $|A_i|$'s for $i \in \mathcal{I}$. 
This is important since, due to \cref{lem: Ai's half the boundary}, it implies that we consider at least half of the nodes in $\d I_t$.

We analyze the above three cases separately in the three following subsections, that will be further subdivided into a total of 7 cases: in 4 cases $I_t$ grows, while in 3 cases $\d I_t$ grows, always with high probability.
By summarizing all of them we get with high probability that either
\[
    |I_{t+r} \setminus I_t| 
    \ge \min\left\{\frac{1}{4}, \frac{1}{20}, \frac{1}{128}, \frac{1}{8}\right\} |A_i| 
    = \frac{|A_i|}{128}
\]
within $r=O(\max\{1, \log_\PPPP n, 1, 1)\}=O(\log_\PPPP n)$ rounds or
\[
    |\d I_{t+r} \setminus \d I_t| 
    \ge \min\left\{\frac{\PPPP}{32 \log n}, \frac{\sqrt{\PPPP}}{4}, \frac{\PPPP}{64}\right\} |A_i|
    \ge \frac{\PPPP^{1/6}}{64}|\d I_t|
\]
within $r=O(1)$ rounds,
since by definition $|A_i| \ge \frac{|\d I_t|}{4 \log n}$ and by assumption $\PPPP > \log^3 n$.

\subsubsection{Sets with low degree nodes}
Recall that in this case every $u \in A_i$ is such that 
\[
    d_i \le 1024 \cdot 96 \cdot \PPPP.
\]

Since the degree of these boundary nodes is low, it is likely that many of them will directly pull the rumor. 
We prove a more general lemma that will be used in other cases as well and formalizes the fact that nodes that are ``well-connected'' to $I_t$ are likely to pull the rumor directly.
In particular we consider nodes of a generic subset $B \subseteq \d I_t$ of the boundary of $I_t$ and bound the number of rounds needed by a constant fraction of nodes in $B$ to pull the rumor from $I_t$, with high probability.

\begin{lemma}\label{lem:lemma 3.5}
Let $B \subseteq \d I_t$ be such that $|B|= \Omega(\log n)$.
Let $0< q\leq 1$ and suppose that for every node $u\in B$ at least a $q$-fraction of its neighbors is in $I_t$. Then with high probability at least $\tfrac{|B|}{4}$ nodes of $B$ have pulled the rumor from $I_t$ in $\lceil \tfrac{1}{q\PPPP}\rceil$ rounds. 
\end{lemma}
\begin{proof}
We consider two different cases, depending on the value of $q$. We show for both cases that the probability of pulling from an informed node is at least $1/2$.

\noindent
\textbf{Case 1.} If $q\PPPP>1$, we show we only need 1 round. The probability that any $u\in B$ pulls from $I_t$ in one round is $1-\left(1-\frac{|N(u)\cap I_t|}{\deg(u)}\right)^\PPPP \geq 1-\left(1-q\right)^\PPPP$. 
Since $q\PPPP>1$, then $1-\left(1-q\right)^\PPPP \geq 1/2$. 

\noindent
\textbf{Case 2.} If $q\PPPP \leq 1$, then the probability that a node pulls the rumor in $\lceil \tfrac{1}{q\PPPP}\rceil$ rounds is $1-\left(1-q\right)^{\PPPP\cdot \frac{1}{q \PPPP}} \geq 1-e^{-1}\geq 1/2$. 

This means that the expected number of nodes that pull in this many rounds is at least $|B|/2$. Since these pulls are independent random variables and $|B|=\Omega(\log n)$, a Chernoff bound gives us that at least $|B|/4$ nodes are informed with high probability. 
\end{proof}

We are now ready to show that the set $A_i$ grows sufficiently.
For each node $u \in A_i$ the fraction of neighbors that $u$ has in $I_t$ is at least $1/\deg(u) \ge 1/(2d_i)$.
Then \cref{lem:lemma 3.5} gives that with high probability a fraction $1/4$ of the nodes in $A_i$ pull the rumor from $I_t$ in at most $r = \lceil 2\cdot 1024\cdot96\cdot\PPPP/\PPPP \rceil = O(1)$ rounds.
By definition of $\mathcal{I}$ in \cref{eq: A_i good index}, with high probability it follows that 
\[
    |I_{t+r} \setminus I_t| \ge \tfrac{1}{4} |A_i|.
\]

\subsubsection{Sets with medium degree nodes}\label{sec:medium}
Recall that in this case every $u \in A_i$ is such that
\[
    1024\cdot96\cdot \PPPP < d_i \le 16 \, |A_i|.
\]

Unlike the previous case, nodes with medium degree can either contribute directly to the growth of $I_t$ or to the growth of $\d I_t$.
To describe these different types of contributions, 
for each node $v \in V$ let us denote the number of neighbors of $v$ in $A_i$ as
\[
    h_i(v) := |N(v) \cap A_i|.
\]
Moreover, let us define the set of nodes that are neither informed or in the boundary of the informed nodes as 
\[
    S_t := V \setminus (I_t \cup \d I_t).
\]

We further distinguish three cases, depending on the \emph{volume} of $A_i$, we define the volume of a set $S$ as $\Vol(S):=\sum_{v\in S} \deg(v)$. 
\begin{enumerate}[i)]
    \item It holds that $\sum_{{v \in S_t \text{ s.t.} h_i(v) \le  d_i\log n/\PPPP}} h_i(v) \ge \frac{1}{2} \Vol(A_i)$.\label{case:medium_i}
    \item It holds that $\sum_{{v \in V \text{ s.t.} h_i(v) \ge d_i \log n/ 96\PPPP}} h_i(v) \ge \frac{1}{3} \Vol(A_i)$.\label{case:medium_ii}
    \item None of the above conditions are met, i.e., it holds that
    $\sum_{{v \in S_t \text{ s.t.} h_i(v) \le d_i\log n/\PPPP}} h_i(v) < \frac{1}{2}\Vol(A_i)$
    and
    $\sum_{{v \in V \text{ s.t.} h_i(v) \ge d_i \log n / 96\PPPP}} h_i(v)
    < \frac{1}{3} \Vol(A_i)$.\label{case:medium_iii}
\end{enumerate}
In the remainder of this subsection we prove that in 
case \ref{case:medium_i}) $\d I_t$ grows, while in cases \ref{case:medium_ii}) and \ref{case:medium_iii}) it is $I_t$ to grow directly.

\paragraph*{Case \ref{case:medium_i})}
It holds that
\begin{equation}
    \sum_{\substack{v \in S_t \text{ s.t.}\\ h_i(v) \le d_i\log n/\PPPP}} h_i(v) \ge \frac{1}{2} \Vol(A_i).
\end{equation}

In this case we give a lower bound on the number of new nodes in the boundary and prove that it holds with high probability. 
We start by noting that the probability that a fixed node $u \in A_i$ pulls the rumor from $I_t$ in a given round is 
\begin{equation*}
    1-\left(1- \frac{|N(u)\cap I_t|}{\deg(u)}\right)^\PPPP 
    \geq 1-\left(1- \frac{1}{\deg(u)}\right)^\PPPP 
    \geq 1-\left(1- \frac{1}{2d_i}\right)^\PPPP 
    \ge \frac{\PPPP}{4d_i}=:p,
\end{equation*}
where we use that $2d_i > \PPPP$. 

Let us pessimistically assume that such a probability exactly equals $p$ for every node $u$.
For each $u\in A_i$, let $X_u$ denote the $0/1$ random variable that is $1$ iff $u$ pulls the rumor from $I_t$ in round $t+1$. 
Then we have $\P[X_u=1]=p$.
Further, for each $v\in S_t$ such that $h_i(v) \leq d_i\log n/\PPPP $, let $Y_v$ denote the $0/1$ random variable which is $1$ exactly when $v$ has a neighbor $u$ with $X_u=1$. 
Finally let $Y= \sum_{v\in V}Y_v$ be the number of new nodes in the boundary only considering the contribution from pull. 
Our goal is to prove a lower bound on $Y$. 

We see that
\begin{align*}
    \P[ Y_v =1] \geq 1-\left(1-\frac{\PPPP}{4d_i}\right)^{h_i(v)} \geq 1-\left(1-\frac{\PPPP}{4d_i}\right)^{h_i(v)/\log n} 
    \geq \frac{\PPPP h_i(v)}{8d_i\log n},
\end{align*}
where the last inequality follows from Taylor series expansion, which needs that $\tfrac{\PPPP h_i(v)}{8d_i\log n}\le 1$, as ensured by the second to last inequality\footnote{Note that without the extra $\log n$ the rightmost term could be bigger than 1, rendering the inequality trivially false.}.
This gives us that in expectation
\begin{align*}
    \E[Y] &= \sum_{\substack{v\in S_t \text{ s.t.}\\ h_i(v) \leq d_i\log n/\PPPP}} \P[ Y_v =1] 
    \geq \sum_{\substack{v\in S_t \text{ s.t.}\\ h_i(v) \leq d_i\log n/\PPPP}} \frac{\PPPP h_i(v)}{8d_i\log n}\\
    &\geq \frac{\PPPP}{16 d_i \log n}\Vol(A_i) 
    \geq \frac{\PPPP}{16 d_i\log n} \cdot d_i|A_i| = \frac{\PPPP}{16\log n}|A_i|.
\end{align*}

Next, we need to show that $Y$ is concentrated around its expectation. A simple Chernoff bound does not suffice: the $Y_v$ are not independent, nor negatively correlated. Instead, we use the method of bounded differences (see \cref{thm:bounded differences}). 
In particular, we apply this theorem with $R_i=X_u$, $f(\{X_u\}_{u\in A_i})=Y$, so $\mu = \E[Y]$ and $b := \max|f(x)-f(x')|\leq 2d_i$. With $\lambda=\frac{\PPPP}{32\log n} |A_i| $ we get
\begin{align*}
    &\P\left[Y < \frac{\PPPP}{32\log n} |A_i|\right] = \P\left[Y < \frac{\PPPP}{16\log n} |A_i|-\frac{\PPPP}{32\log n} |A_i|\right] \\ 
    &\le \exp\left(-\frac{\PPPP |A_i|}{32^2(2+1/24) \cdot d_i\log^2 n}\right)
    \leq \exp\left(-\frac{\PPPP |A_i|}{2091 \cdot d_i\log^2 n}\right).
\end{align*}
Using that $d_i\leq 16 |A_i|$, and that $\PPPP > \log^3 n$, we get with high probability that
\[
    |\d I_{t+1} \setminus \d I_t| 
    \ge \frac{\PPPP}{32\log n}|A_i|.
\]

\paragraph*{Case \ref{case:medium_ii})}
It holds that
\begin{equation}
    \sum_{\substack{v \in V \text{ s.t.}\\ h_i(v) \ge d_i \log n/ 96\PPPP}} h_i(v) \ge \frac{1}{3} \Vol(A_i).
\end{equation}
Intuitively, we want to formalize the fact that the number of informed nodes in $A_i$ must grow even if the nodes are not likely to pull directly from $I_t$. We do this by looking at a subset $A_i'$ of the nodes in $A_i$ that have many connections. More formally we have the following fact, for a proof we refer to~\cite{GiakkoupisS12}. 

\begin{claim}[Section 3.4 in \cite{GiakkoupisS12}]\label{claim:someting_with_Ai}
    Let $\ell=d_i \log n/ 96\PPPP$. 
    Then there exist sets $A_i'\subseteq A_i$ and $V'\subseteq V$ such that each node in $A_i'$ has at least $d_i/8$ neighbors in $V'$ and each node in $V'$ has at least $\ell/8$ neighbors in $A_i'$, and the size of $A_i'$ is at least $|A_i'|\geq |A_i|/20$.
\end{claim}

We will show that each node of $s\in A_i'$ gets informed in $O(\log_\PPPP n)$ rounds with high probability. By \Cref{lm:symmetry}, this is the same as a rumor from $s$ reaching $I_t$ after $O(\log_\PPPP n)$ rounds. First, we show that a rumor started at $s\in B$ reaches $\ell/1024$ nodes in $A_i'$ in $O(\log_\PPPP n)$ rounds w.h.p. Note that $\ell/1024\geq 1$ since $d_i \geq 1024\cdot 96 \PPPP$.
Then the probability that one of these nodes pushes the rumor to $I_t$ in the next $c\cdot 2048\cdot 96$ rounds is at least
\begin{align*}
    1-\left(1-\frac{1}{2d_i}\right)^{\PPPP \frac{\ell}{1024}c\cdot 2048\cdot 96}=  1-\left(1-\frac{1}{2d_i}\right)^{c\cdot2d_i\log n}\ge 1-n^{-c},
\end{align*}
where the equality follows from the definition of $\ell$.

To show that a rumor from $s\in A_i'$ reaches $\ell/1024$ nodes in $A_i'$ in $O(\log_\PPPP n)$ rounds, we show that in $\log_\PPPP n$ phases of $O(1)$ rounds, the number of informed nodes in $A_i'$ grows by a factor $\PPPP$. We do this in two steps. First we show that by push $A_i'$ informs nodes in $V'$, and then by pull nodes from $A_i'$ pull the rumor from $V'$.

Let $m$ denote the number of informed nodes of $A_i'$ at the start of the phase. We show that by $O(1)$ push rounds we have at least $\min\{\PPPP m,\ell/16\}$ informed nodes in $V'$. Suppose we have less than  $\min\{\PPPP m,\ell/16\}$ informed nodes $V'_i$ in $V'$ after $i$ pushes. 
Then the probability of a successful push from $u\in A_i'$ is at least
\begin{align*}
    \frac{|(N(u)\cap V')\setminus V'_i|}{\deg(u)} &\geq \frac{|N(u)\cap V'|-| V'_i|}{\deg(u)}\geq  \frac{d_i/8 - \min\{\PPPP m,\ell/16\}}{2d_i}\\
    &\geq \frac{d_i/8 -d_i/16}{2d_i} \geq \frac{1}{32}.
\end{align*}
So the expected number of informed nodes in $V'$ after 1 round is at least $\min\left\{\tfrac{m \PPPP}{32},\ell/16\right\}$.
Since these pushes are negatively correlated, Chernoff, \Cref{thm:chernoff bound}, gives that we have with high probability that the number of informed nodes in $V'$ after 1 round is at least 
\begin{equation*}
    \min\left\{\frac{m \PPPP}{64},\ell/16\right\}.
\end{equation*}

Now we consider pulling from $V'$. Suppose $m'$ nodes $V'_{\text{inf}}$ in $V'$ are informed, now the expected number of informed nodes in $A_i'$ after $96$ pull rounds is at least
\begin{align*}
    &\sum_{u\in A_i'} 1-\left(1-\frac{|N(u)\cap V'_{\text{inf}}|}{\deg(u)}\right)^{96\PPPP } 
    = \sum_{u\in A_i'} 1-\left(1-\frac{|N(u)\cap V'_{\text{inf}}|}{\deg(u)}\right)^{d_i \log n/\ell }\\
    \geq &\sum_{u\in A_i'} 1-\left(1-\frac{|N(u)\cap V'_{\text{inf}}|}{2d_i}\right)^{\tfrac{d_i}{\ell}} 
    \geq \sum_{u\in A_i'}\frac{d_i}{\ell}\frac{|N(u)\cap V'_{\text{inf}}|}{4d_i}\\
    = &\frac{1}{4\ell}\sum_{v\in V'_{\text{inf}}}|N(v)\cap A_i'|
    \geq \frac{1}{4\ell}\sum_{v\in V'_{\text{inf}}}\frac{\ell}{8}
    = \frac{|V'_{\text{inf}}|}{32}, 
\end{align*}
where the second inequality uses that $|N(u)\cap V'_{\text{inf}}|\leq |V'_{\text{inf}}|\leq \ell/16$. So we conclude that in expectation, we have at least $\min\left\{\frac{m \PPPP}{32\cdot 64},\frac{\ell}{32\cdot 16}\right\}$ informed nodes in $A_i'$.
Since all pulls are independent, Chernoff, \Cref{thm:chernoff bound}, gives us that with high probability we have at least 
\begin{equation*}
    \min\left\{\frac{m \PPPP}{4096},\frac{\ell}{1024}\right\}.
\end{equation*}
informed nodes in $A_i'$.

We conclude that in $r=O(\log_\PPPP n)$ rounds we have 
\begin{align*}
    |I_{t+r}\setminus I_t| \geq |A_i'|\geq \frac{|A_i|}{20},
\end{align*}
where the last inequality comes from \Cref{claim:someting_with_Ai}.

\paragraph*{Case \ref{case:medium_iii})}
None of the previous conditions holds, i.e., we have that
\[
    \sum_{\substack{v \in S_t \text{ s.t.}\\ h_i(v) \le d_i \log n/\PPPP}} h_i(v) < \frac{1}{2}\Vol(A_i)
    \qquad\text{and}\qquad
    \sum_{\substack{v \in V \text{ s.t.}\\ h_i(v) \ge d_i  \log n/ 96\PPPP}} h_i(v)
    < \frac{1}{3} \Vol(A_i).
\]

By using the two above conditions on the volume of $A_i$, we can see that the number of edges going from $A_i$ to some $v \in V \setminus S_t$ with $h_i(v) < d_i / 96\PPPP$ is at least a constant fraction of the volume of $A_i$, namely
\begin{align*}
    \sum_{\substack{v \in V \setminus S_t \text{ s.t.}\\ h_i(v) < d_i  \log n/ 96\PPPP}} h_i(v) 
    &= \sum_{v\in V} h_i(v) 
    - \sum_{\substack{v\in V \text{ s.t.}\\ h_i(v)\ge d_i  \log n/ 96\PPPP}} h_i(v) 
    - \sum_{\substack{v\in S_t \text{ s.t.}\\ h_i(v) < d_i  \log n/ 96\PPPP}} h_i(v)
    \\
    &> \Vol(A_i) 
    - \sum_{\substack{v\in V \text{ s.t.}\\ h_i(v)\ge d_i  \log n/ 96\PPPP}} h_i(v) 
    - \sum_{\substack{v\in S_t \text{ s.t.}\\ h_i(v) \le d_i  \log n/\PPPP}} h_i(v)
    \\ 
    &\geq \frac{1}{6} \Vol(A_i).
\end{align*}
Moreover, the number of such edges going to $\d I_t$ is at most $\tfrac{d_i \log n}{96\PPPP} |\d I_t|$, since $h_i(v) < d_i \log n/96\PPPP$.
Therefore, using the condition on $|A_i|$ in \cref{eq: A_i good index}, that $\deg(u) \ge d_i$ for every $u \in A_i$, and that $\PPPP>\log^2 n$, we get
\begin{equation}\label{eq: lb on cut Ai/It}
    |E(A_i,I_t)| \ge \frac{1}{6}\Vol(A_i) - \frac{d_i \log n}{96\PPPP} |\d I_t|
    \ge \frac{d_i |A_i|}{6} - \frac{d_i \log n}{96\PPPP} (4\log n \cdot |A_i|)
    \ge \frac{d_i |A_i|}{8}.
\end{equation}

Let $B \subseteq A_i$ be the set of nodes in $A_i$ that have at least $d_i/16$ neighbors in $I_t$.
Then, using that $\deg(u) < 2d_i$ for every $u \in A_i$, we get
\begin{equation}\label{eq: ub on cut Ai/It}
    |E(A_i,I_t)| \le |B| \cdot 2d_i + (|A_i|-|B|) \frac{d_i}{16}.
\end{equation}
By combining \cref{eq: lb on cut Ai/It,eq: ub on cut Ai/It} it follows that
\(
    |B| \ge \frac{|A_i|}{32}.
\)

Note that each node $u \in B$ in the set has a fraction of at least $\frac{d_i/16}{\deg(u)} \ge \frac{d_i/16}{2d_i} = 1/32$ neighbors in $I_t$.
We can then apply \cref{lem:lemma 3.5} and get that with high probability at least $|B|/4 \ge |A_i|/128$ nodes in $A_i$ pull the rumor from $I_t$ in at most $\lceil 1/\lceil  \PPPP/32 \rceil\rceil = 1$ round, since $\PPPP>\log^2 n$.
Hence with high probability it follows that
\[
    |I_{t+1} \setminus I_t| 
    \ge \frac{|A_i|}{128}.
\]

\subsubsection{Sets with high degree nodes}\label{sec:high}
Recall that in this case every $u \in A_i$ is such that
\[
    d_i \ge 2 |\d I_t|.
\]

As in the previous case, also here nodes can either contribute directly to the growth of $I_t$ or to that of the boundary $\d I_t$.
Recall that $S_t := V \setminus (I_t \cup \d I_t)$.
Let us denote the set of nodes in $A_i$ that have more neighbors toward $S_t$ than $I_t$ as
\[
    B := \{ u \in A_i : |N(u) \cap S_t| \ge \PPPP|N(u) \cap I_t| \}.
\]
We distinguish three cases:
\begin{enumerate}[i)]
    \item $|B| \ge \frac{1}{2}|A_i|$, and $\deg(u)\geq \sqrt{\PPPP}|B|$. \label{case:high_i}
    \item $|B| \ge \frac{1}{2}|A_i|$, $\deg(u)< \sqrt{\PPPP}|B|$.\label{case:high_ii}
    \item $|B| < \frac{1}{2}|A_i|$. \label{case:high_iii}
\end{enumerate}
As before in the remainder of this subsection we prove separately for each of these cases that either $I_t$ or $\d I_t$ grows sufficiently. 
In particular in cases \ref{case:high_i}) and \ref{case:high_ii}) the boundary grows sufficiently, while in case \ref{case:high_iii}) it is $I_t$ to grow.

\paragraph*{Case \ref{case:high_i})} 
It holds that
\[
    |B| \ge \frac{1}{2}|A_i| 
    \quad\text{and}\quad
    \deg(u)\geq \sqrt{\PPPP}|B|.
\]

In this case, we start by showing that after $\lceil \frac{|I_t|+|\d I_t|}{|B| \log n}\rceil$ rounds, we inform at least one node $u\in B$ with high probability. 

The probability that a fixed node $v\in I_t$ does not push to any neighbor in $B$ is $1-\tfrac{h_i(v)}{\deg(v)}\leq 1-\tfrac{h_i(v)}{|I_t|+|\d I_t|}$. 
Hence, the probability that no $v\in I_t$ pushes to $B$ within $\frac{|I_t|+|\d I_t|}{|B|\log n}$ rounds of \pushpull{$\PPPP$} is upper bounded by
\[
    \prod_{v\in I_t}\left(1-\frac{h_i(v)}{|I_t|+|\d I_t|}\right)^{\PPPP \frac{|I_t|+|\d I_t|}{|B|\log n}} \leq e^{-\PPPP/\log n \frac{1}{|B|} \sum_{v\in I_t}h_i(v)} \leq e^{-\PPPP/\log n}.
\]
This uses that $\sum_{v\in I_t}h_i(v)= \sum_{v\in A_i}|N(u)\cap I_t|\geq |A_i| \geq |B|$, since $A_i\subseteq \d I_t$.

Hence, the probability that push informs at least one node in these rounds is at least $1-e^{-\PPPP/\log n}$, namely with high probability, since $\PPPP > \log^2 n$.

Therefore, since $\deg(u)\geq \sqrt{\PPPP}|B|$, and by definition of the set $B$, we get with high probability that
\[
    |\d I_{t+r} \setminus \d I_t |
    \ge \frac{\sqrt{\PPPP}|B|}{2}
    \ge \frac{\sqrt{\PPPP}|A_i|}{4}
    \ge \frac{\sqrt{\PPPP}|\d I_t|}{16 \log n}
\]
for a number of rounds $r=\lceil \frac{|I_t|+|\d I_t|}{|B| \log n}\rceil=O(1)$ since $|B|\geq \tfrac{1}{2}|A_i|\geq \tfrac{|\d I_t|}{8\log n}$, where $|A_i|\geq \tfrac{|\d I_t|}{4\log n}$ since we consider $i\in \mathcal I$.

\paragraph*{Case \ref{case:high_ii})} 
It holds that
\[
    |B| \ge \frac{1}{2}|A_i|
    \quad\text{and}\quad
    \deg(u)< \sqrt{\PPPP}|B|.
\]

We start by noting that the probability that a fixed node $u \in A_i$ pulls the rumor from $I_t$ in a given round is 
\begin{equation*}
    1-\left(1- \frac{|N(u)\cap I_t|}{\deg(u)}\right)^\PPPP 
    \geq 1-\left(1- \frac{1}{\deg(u)}\right)^\PPPP 
    \geq 1-\left(1- \frac{1}{2d_i}\right)^\PPPP 
    \ge \frac{\PPPP}{4d_i}=:p,
\end{equation*}
where we use in the last inequality that $2d_i > \PPPP$. 

Formally, let $X_u$ be the 0/1 random variable that is 1 iff $u$ pulls the rumor. Note that the $X_u$'s are independent since they come from pull operations. 
Let us pessimistically assume that such a probability exactly equals $p$ for every node $u$.
For each node $v \in N(B) \cap S_t$ let $Y_v$ be the 0/1 random variable that is 1 iff $v$ has at least a neighbor $u$ with $X_u=1$. 

We look at $v\in N(B)\cap S_t$ and consider the probability that no neighbor of $v$ in $B$ pulls the rumor, which is at least
\begin{equation*}
    1-\prod_{u\in B\cap N(v)}\left(1-\frac{\PPPP}{4d_i}\right) 
    \geq 1-\exp\left(-\PPPP/4\sum_{u\in B\cap N(v)}\frac{1}{d_i}\right)
    \geq \frac{\PPPP}{8} \sum_{u\in B\cap N(v)}\frac{1}{d_i},
\end{equation*}
where the last inequality follows from the assumption on nodes $u\in B$. 
The expected number of nodes $v\in N(B)\cap S_t$ that join $\d I_t$ is now at least
\begin{equation*}
    \sum_{v\in N(B)\cap S_t} \frac{\PPPP}{8} \sum_{u\in B\cap N(v)}\frac{1}{d_i} 
    = \frac{\PPPP}{8}\sum_{u\in B}\frac{|N(u)\cap S_t|}{d_i}
    \geq \frac{\PPPP}{16}|B|,
\end{equation*}
where the second to last equality follows by definition of $B$ and the fact that $d_i>2|\d I_t|$.

We now show that the expected number of new nodes in the boundary is not far from that lower bound. Formally we  lower bound $Y:= \sum_{v\in N(B)\cap S_t} Y_v$.
Hereto, we use the method of bounded differences (\cref{thm:bounded differences}) with $R_i = X_u$, $f(\{X_u\}_{u \in B})=Y$, and $b \le 2d_i$.
Since $\E[Y] \ge \frac{\PPPP}{16}|B|$, with $\lambda = \frac{\PPPP}{32}|B|$ we get
\[
    \P\left(Y < \frac{\PPPP}{32}|B|\right) \le \exp\left(-\frac{\PPPP |B|}{2091 \cdot d_i}\right).
\]
Since $\deg(u)<2d_i < 2\sqrt{\PPPP}|B|$ and $\PPPP > \log^2 n$, the previous bound holds with high probability.
Therefore with high probability we have
\[
    |\d I_{t+1} \setminus \d I_t| \ge \frac{\PPPP}{64}|A_i|.
\]

\paragraph*{Case \ref{case:high_iii})}
It holds that 
\[
    |B| < \frac{1}{2}|A_i|.
\]

It follows directly from the above condition that $|A_i \setminus B| > |A_i|/2$ and by the definition of $B$ that each node $u \in A_i \setminus B$ iff $|N(u)\cap S_t|< k |N(u)\cap I_t|$. Recall that we are in the case of high degree nodes: $|N(u)|\ge 2|\d I_t|$, equivalently, we have $|\d I_t|\le \tfrac{1}{2}|N(u)|$. Hence $|N(u)\cap \d I_t| \le |\d I_t|\le \tfrac{1}{2}|N(u)|$, and $|N(u)\cap I_t|+|N(u)\cap S_t|\ge \tfrac{1}{2}|N(u)|$, so $|N(u)\cap \d I_t| \le|N(u)\cap I_t|+|N(u)\cap S_t|$. Combining this, we obtain
\begin{align*}
    &\frac{|N(u)\cap I_t| }{|N(u)|} 
    = \frac{|N(u)\cap I_t| }{|N(u)\cap I_t|+|N(u)\cap S_t|+ |N(u)\cap \d I_t|} \\
    &\ge \frac{|N(u)\cap I_t| }{2(|N(u)\cap I_t|+|N(u)\cap S_t|)} 
    > \frac{|N(u)\cap I_t| }{2(k+1)|N(u)\cap I_t|}= \frac{1}{2(k+1)}.
\end{align*}

We conclude that $u$ has a $\tfrac{1}{2(k+1)}$-fraction of all its neighbors in $I_t$.
Then we apply \cref{lem:lemma 3.5} and get that with high probability at least one fourth of the nodes in $|A_i \setminus B|$ pull the rumor from $I_t$ in at most $\lceil \tfrac{2(k+1)}{k}\rceil=O(1)$ rounds.
Hence we get with high probability that
\[
    |I_{t+r} \setminus I_t|
    \ge \frac{|A_i \setminus B|}{4}
    \ge \frac{|A_i|}{8}
\]
for a number of rounds $r=O(1)$.

 \newpage
\section{On Small-Set Vertex Expanders with Expansion Larger than 1}\label{sec:expanders}

As previously discussed, when the parameter $\alpha < 1/2$ we talk about small-set expanders.
These graphs are substantially different from the classical notion of expanders, where $\alpha = 1/2$.
In fact, if $\alpha = 1/2$ it holds that the vertex-expansion of the graph is bounded, having $\phi \le 1$.
The following lemma formalizes this fact. 

\begin{lemma}\label{lm:Imp_Exp}
    Let $\phi> 0$. There are no $(\phi,\alpha)$-expanders for $\alpha>\tfrac{1}{1+\phi}$.
\end{lemma}
\begin{proof}
     Let $S$ be any set of size $\tfrac{n}{1+\phi}< |S| \leq \alpha n$. We demand that $\phi(S)\geq \phi$, in other words:
\[
    |S\cup N(S)| = |S|+|\d S| \geq |S| + \phi|S|>(1+\phi)\tfrac{n}{1+\phi}=n,
\]
which clearly is not possible. 
\end{proof}

Now the follow-up question is: why do we not simply take $\alpha=\tfrac{1}{1+\phi}$? The answer is that this is a very restricted graph class as we show next. 
Recall that the neighborhood of $S$ is defined as $N(S):=  \{v\in V:\exists s\in S,\,\{s,v\}\in E\}$. We now define the \emph{inclusive neighborhood} of a set $S$ as $N[S]:= S\cup N(S)$. 
We define the $i$-th neighborhood as $N^{i}[v]:=\{ u\in V: d(u,v)\leq i\}=N[N^{i-1}[v]]$ where $d(u,v)$ is the shortest-path distance between $u$ and $v$. 

\begin{lemma}\label{lm:Restricted_Exp}
    Let $\phi\geq 1$. If $G$ is a $\left(\phi,\tfrac{1}{1+\phi}\right)$-expander, then $G$ has diameter at most $2$ and minimum degree $\tfrac{\phi}{1+\phi}n$.
\end{lemma}
\begin{proof}
    Let $D$ denote the diameter of $G$. Let $u,v\in V$ be a pair of nodes with distance $D$. Then we have that $u\in N^D[v]\setminus N^{D-1}[v]$. That means that $N^{D-1}[v]\neq V$, so $|N^{D-2}[v]|< \tfrac{n}{1+\phi}$, and $|N^D[v] \setminus N^{D-1}[v]|< \tfrac{n}{1+\phi}$. That means that
    \begin{align*}
        |N^{D-1}[v]\setminus N^{D-2}[v]|&=n- |N^{D-2}[v]|-|N^{D}[v]\setminus N^{D-1}[v]|
        > n-2\tfrac{n}{1+\phi}
        =(\phi-1) \tfrac{n}{1+\phi}.
    \end{align*}
    But that means that the nodes at distance $D-1$, i.e., $N^{D-1}[v]\setminus N^{D-2}[v]$, have edges to the entire graph, so also to $v$. Hence $N^{D-1}[v]\setminus N^{D-2}[v]= N[v]\setminus v$, and thus $D=2$.

   For the minimum degree we now note that
   \[
       \deg(v) = |N(v)|= n- |N^{2}[v]\setminus N [v]|-1\geq n-\tfrac{n}{1+\phi}=\tfrac{\phi}{1+\phi}n, 
   \]
   since $|N^{2}[v]\setminus N[v]|< \tfrac{n}{1+\phi}$, so $|N^{2}[v]\setminus N[v]|+1 \leq \tfrac{n}{1+\phi}$.
\end{proof}

So to make the definition less restrictive, we allow for $\alpha$ to take more values. However, for small values of $\alpha$, the expansion property becomes local, which does not help us for rumor spreading. 

\begin{lemma}\label{lm:Disc_Exp}
    Let $\alpha\leq \tfrac{1}{2+2\phi}$. For any even $n\in \N$, there exists a $(\phi,\alpha)$-expander $G$ that is not connected. 
\end{lemma}
\begin{proof}
    Let $G=K_{n/2}+K_{n/2}$ be the sum of two disjoint complete graphs, which is clearly not connected. 
    We show that $G$ is a $(\phi,\alpha)$-expander. 
    Let $S\subset V$ with $|S|\leq \alpha n$. We have 
    \begin{equation*}
        |\d S| \geq \frac{n}{2}-|S| \geq \frac{n}{2}-\alpha n 
        \ge \frac{n}{2}-\frac{n}{2+2\phi}
        = \frac{n}{2}\left(1-\frac{1}{1+\phi}\right) 
        = \frac{\phi n}{2+2\phi}
        \geq \phi \alpha n \geq \phi |S|. 
        \qedhere
    \end{equation*}
\end{proof}

We conclude that for rumor spreading we can focus on the regime $\tfrac{1}{2+2\phi}< \alpha \leq \tfrac{1}{1+\phi}$. We show that in this regime a $(\phi,\alpha)$-expander has small diameter. 

\begin{lemma}\label{lm:Exp_Diam}
    If $G$ is a $(\phi,\alpha)$-expander for $\alpha> \tfrac{1}{2+2\phi}$, then it has diameter at most $O(\log_{\phi}n)$.
\end{lemma}
\begin{proof}
    Let $v\in V$ be an arbitrary node, we show that $|N^{\lceil\log_{\phi}n\rceil}[v]|>  n/2$, hence for any pair of nodes $u,v\in V$, we have $N^{\lceil\log_{\phi}n\rceil}[u]\cap N^{\lceil\log_{\phi}n\rceil}[v]\neq \emptyset$, so there is a path from $u$ to $v$ of length at most $2\lceil\log_\phi n\rceil$. 

    By expansion we have that $|N^{i}[v]|\geq \phi^i$, if $|N^{i-1}[v]|\leq \alpha n$.
    Now if $|N^{\lceil\log_{\phi}n\rceil}[v]|\geq \phi^{\log_{\phi}n}=n$, we are done. 
    So suppose not, i.e., suppose $|N^{\lceil\log_{\phi}n\rceil-1}[v]|>  \alpha n$. 
    For every $S \subset N^{\lceil\log_{\phi}n\rceil-1}[v]$ with $|S|=\alpha n$ it holds that $|S \cup \d S| \ge (1+\phi) \cdot \alpha n$ since $G$ is a $(\phi, \alpha)$-expander.
    Hence $|N^{\lceil\log_{\phi}n\rceil}[v]|\ge |S \cup \d S| \geq  (1+\phi) \cdot \alpha n$.
    Since we assume $\alpha> \tfrac{1}{2+2\phi}$ we have that
    \(
        (1+\phi)\alpha n > \frac{1+\phi}{2+2\phi}n=\frac{n}{2},
    \)
    which completes the proof. 
\end{proof}

\subsection{Examples of Small-Set Vertex Expanders}
A trivial example of $(\phi,\alpha)$-expanders are complete graphs. 
\begin{lemma}\label{lm:Complete_Exp}
The complete graph on $n$ nodes is a $(\phi,\alpha)$-vertex expander for every $0< \phi \leq n-1$ and $\alpha \leq \tfrac{1}{1+\phi}$.
\end{lemma}
\begin{proof}
Let $S\subseteq V$ with $|S|\leq \alpha n$. Then we have
\[
    \frac{|\d S|}{|S|}= \frac{n-|S|}{|S|}= \frac{n}{|S|}-1 \geq \frac{1}{\alpha} -1\geq \phi. 
    \qedhere
\]
\end{proof}

However, we can show also the existence of significantly sparser graphs that are $(\phi,\alpha)$-expanders. We consider Erd\H{o}s-Rényi random graphs $G(n,p)$ on $n$ nodes, where there is an edge between any pair of nodes with probability $p$. 
We also know that w.h.p.\ the graphs $G(n,\phi/n)$ have diameter $\Theta(\log_\phi n)$ for $\phi=\omega(1)$~\cite{ChungL01}.

\begin{lemma}\label{lm:randomgraph}
    Let $\alpha< \tfrac{1}{1+\phi/(1-e^{-1})}\approx \tfrac{1}{1+1.58\phi}$. Let $a$ be such that $\alpha = \tfrac{1}{1+a\phi}$ and $\phi\geq \Theta\left(\left(\tfrac{1}{a(1-e^{-1})-1}\right)^2  \log n\right)$. Erd\H{o}s-Rényi random graphs $G(n,p)$ with $p\geq \tfrac{3\phi}{n}$ are $(\phi,\alpha)$-vertex expanders with high probability.
 \end{lemma}
\begin{proof}
    Let $S\subseteq V$ with $|S|\leq \alpha n=\tfrac{n}{1+a\phi}$. 
    Consider $u\in V\setminus S$. 
    We have that
    \begin{equation*}
        \P[u\in N(S): u\in \overline S] = 1-(1-p)^{|S|}.
    \end{equation*}
    We consider two different cases.
    
    If $|S|\leq \tfrac{1}{p}$, then we know $1-(1-p)^{|S|}= p|S|-\tfrac{1}{2}|S|(|S|-1)p^2+ \dots\geq \tfrac{1}{2}p|S|$, and 
    \begin{align*}
        \E[|\d S|]&= |\overline{S}| \P[u\in N(S): u\in \overline S]\\
        &\geq (n-|S|) \tfrac{1}{2}p|S|
        \geq (n-\tfrac{1}{p}) \tfrac{1}{2}p|S|
        = (n- \tfrac{n}{3\phi})\tfrac{1}{2}\tfrac{3\phi}{n}|S|
        =(\tfrac{3}{2}-\tfrac{1}{2\phi})\phi |S|
        \geq \tfrac{5}{4}\phi |S|.
    \end{align*}
    Now we apply a Chernoff bound to get that $\P[|\d S|< \phi |S|]\leq e^{-(1-4/5)^2\phi |S|/2}= e^{-\frac{1}{50}\phi |S|}$ in this case. 

    If $|S|\geq \tfrac{1}{p}$, we use a different approximation. We note that $1-(1-p)^{|S|} \geq 1-e^{-p|S|} $, so 
    \begin{align*}
        \E[|\d S|]&= |\overline{S}| \P[u\in N(S): u\in \overline S]
        \geq (n-|S|) (1-e^{-p|S|})\\
        &\geq (n-\tfrac{n}{1+a\phi})(1-e^{-p|S|})
        \geq a(1-e^{-p|S|})\phi \tfrac{n}{1+a\phi}
        \geq a(1-e^{-1}) \phi |S|.
    \end{align*}

    Now we apply a Chernoff bound to get that 
    \[
        \P[|\d S|< \phi |S|]\leq \exp\left(-\left(1-\frac{1}{a(1-e^{-1})}\right)^2a(1-e^{-1})\phi |S|/2  \right).
    \]

    Using that 
    \[
        \phi > \max\left\{
            50,\frac{2}{\big(1-\tfrac{1}{a(1-e^{-1})}\big)^2a(1-e^{-1})} 
        \right\}(c+2)\log n,
    \]
    we see that in both cases we have 
    \begin{equation*}
        \P[|\d S|< \phi |S|] \leq e^{-(c+2)\log n|S|}.
    \end{equation*}
    
    Now by summing over all choices of $S$ we get
    \begin{align*}
         \sum_{|S|=1}^{n-1}  \binom{n}{|S|} e^{-(c+2)\log n|S|} &\leq \sum_{|S|=1}^{n-1} e^{\log n|S|-(c+2)\log n|S|}
        \leq n e^{-(c+1)\log n}
        =n^{-c},
    \end{align*}
    which completes the proof. 
\end{proof}


 \newpage
\printbibliography[heading=bibintoc]

\appendix
\newpage
\section{Useful Inequalities}
\subsection{Mathematical Approximations}\label{sc:inequalitites}

In the paper we often use the following inequalities:
\begin{itemize}
    \item $1-x \le e^{-x} \le 1-\tfrac{x}{2}$, where the first inequality holds for every $ x\in \mathbb R$ while the second for $x\in[0,1.59]$; the inequalities follow from Taylor series expansion of $e^{-x}$.
    \item $\left(1-\tfrac{1}{x}\right) e^{-\tfrac{k}{x}}\le \left(1-\tfrac{1}{x}\right)^k \le e^{-\tfrac{k}{x}}$, where the inequalities hold for $ x>1$ and for any $ k\in \mathbb R$; the inequalities follow from binomial expansion.
\end{itemize}

\subsection{Concentration Bounds}
\label{sc:tail_bounds}
Here we present some well-known definitions and concentration inequalities that we use throughout the paper and whose proofs can be found in, e.g.,~\cite{dubhashi2009concentration}.

We start with the definition of negatively associated random variables.
\begin{definition}[Negatively Associated Random Variables]\label{def:neg ass rvs}
The random variables $X_1,\dots,X_n$ are said to be negatively associated if for all disjoint subsets $I,J \subseteq \{1,\dots,n\}$ and all non-decreasing functions $f$ and $g$ it holds that
\[
    \E[f(X_i, i\in I) \cdot g(X_j, j\in J)] 
    \le \E[f(X_i, i\in I)] \cdot \E[g(X_j, j\in J)].
\]
\end{definition}

In particular, in many proofs we will face situations that boil down to the following simple fact.
\begin{claim}\label{claim:neg_as}
Let $X,Y$ be $0/1$ random variables. 
If $\P[X = 1 \mid Y=1] \le \P[X = 1]$ then $X,Y$ are negatively associated.
\end{claim}
\begin{proof}
By using the law of total expectation it follows that
\begin{align*}
    \E[XY] &= \E[XY \mid Y=1] \, \P(Y=1) + \E[XY \mid Y=0] \, \P[Y=0]
    \\
    &= (1 \cdot \P[X=1 \mid Y=1] + 0 \cdot \P[X=0 \mid Y=1] ) \, \P[Y=1]
    \\
    &\le \P[X=1] \, \P[Y=1]
    = \E[X] \, \E[Y],
\end{align*}
which proves the statement by definition of negatively associated random variables.
\end{proof}

Next, we present some standard Chernoff's bounds. 
Note that they are extended also to the case of negatively associated random variables.
\begin{theorem}[Chernoff Bounds \cite{dubhashi2009concentration}]\label{thm:chernoff bound}
Let $X_1,\dots,X_n$ be independent $0/1$ random variables. 
Let $X := \sum_{i=1}^{n} X_i$ and let $\mu_L \le \E[X] \le \mu_H$.
It holds that
\begin{align*}
    &\P[ X < (1-\delta) \mu_L ] \le e^{-\frac{\delta^2}{2}\mu_L}, 
    \quad\forall \delta \in (0,1);
    \\
    &\P[ X > (1+\delta) \mu_H ] \le e^{-\frac{\delta^2}{2+\delta}\mu_H}, 
    \quad\forall \delta>0.
\end{align*}

\noindent
The bounds also hold if $X_1,\ldots,X_n$ are negatively associated.
\end{theorem}

Next, we state a result following from McDiarmid's inequality~\cite{mcdiarmid1998concentration}.
\begin{theorem}[Bounded Differences Inequality \cite{mcdiarmid1998concentration}]\label{thm:bounded differences}
Let $R_1, \dots R_n$ be independent $0/1$ random variables with $\P[R_i=1] \leq p \leq 1/2$. Let $f$ be a bounded real function defined on $\{0,1\}^n$. 
Define $\mu := \E[f(R_1, \dots, R_n)]$, and $b:= \max|f(x)-f(x')|$, where the maximum is over all $x,x'\in \{0,1\}^n$ that differ only in one position. Then for any $\lambda>0$
\[
    \P[f(R_1, \dots, R_n)\leq \mu- \lambda] \leq \exp\left(-\frac{\lambda^2}{2pnb^2+2b\lambda/3}\right).
\]
\end{theorem}

Finally, we state a result specifically for the tail of the binomial distribution~\cite{arratia1989tutorial}. 
\begin{theorem}[\cite{arratia1989tutorial}]\label{thm:binom_tail}
    Let $B(n,p)$ denote a random variable following the binomial distribution with $n$ trials and success probability $p$. Then, for $p<a<1$, we have that 
    \begin{equation*}
        \P[B(n,p) \geq an] \leq \exp\left(- n\left(a\log\left(\frac{a}{p}\right)+(1-a)\log\left(\frac{1-a}{1-p}\right)\right)\right).
    \end{equation*}
\end{theorem}

\end{document}